\documentclass{article}
\RequirePackage{amsmath}
\RequirePackage[noblocks]{authblk}
\RequirePackage{natbib}
\RequirePackage{hyperref} 
\RequirePackage{amssymb}
\RequirePackage{graphicx}
\usepackage{subcaption}
\RequirePackage{fullpage}
\RequirePackage{fancyhdr}
\usepackage{hhline}
\usepackage{amsthm}
\usepackage{tabularx,booktabs}
\usepackage[super]{nth}
\usepackage{algorithm2e,algorithmic}
\usepackage{comment}
\usepackage{color}
\usepackage{array}
\usepackage{arydshln}
\usepackage{chngcntr}

\newtheorem{theorem}{\bf Theorem}

\newtheorem{proposition}{\bf Proposition}

\newcommand{\defeq}{\mathrel{\mathop:}=}

\def\R{\texttt{R}}
\def\iid{\overset{\text{iid}}{\sim}}
\def\limma{\texttt{limma}}
\def\voom{\texttt{voom}}
\def\edgeR{\texttt{edgeR}}
\def\EbayesThresh{\texttt{EbayesThresh}}
\def\REBayes{\texttt{REBayes}}
\def\Rmosek{\texttt{Rmosek}}
\def\deconvolveR{\texttt{deconvolveR}}
\def\cash{\texttt{cashr}}
\def\ash{\texttt{ashr}}
\def\qvalue{\texttt{qvalue}}
\def\locfdr{\texttt{locfdr}}

\begin{document}

\title{Solving the Empirical Bayes Normal Means Problem with Correlated Noise}

\author[1]{Lei Sun}
\author[1, 2]{Matthew Stephens\thanks{Contact: Matthew Stephens (\href{mailto:mstephens@uchicago.edu}{\texttt{mstephens@uchicago.edu}}).
    This work is partially supported by National Institutes of Health Grant HG002585 to M. S.}}
\affil[1]{Department of Statistics, University of Chicago}
\affil[2]{Department of Human Genetics, University of Chicago}

\date{}


\maketitle

\begin{abstract}
The Normal Means problem plays a fundamental role in many areas of
modern high-dimensional statistics, both in theory and practice. And the Empirical Bayes (EB) approach to solving this problem has been shown to be highly effective, again both in theory and practice. However, almost all EB treatments of the Normal Means problem assume that the observations are independent. In practice correlations are ubiquitous in real-world applications, and these correlations can grossly distort EB estimates. Here, exploiting theory from \cite{schwartzman2010},
we develop new EB methods for solving the Normal Means problem
that take account of {\it unknown} correlations among observations. We provide practical software implementations of these methods, and illustrate them in the context of large-scale multiple testing problems and False Discovery Rate (FDR) control. In realistic numerical experiments our methods compare favorably with other commonly-used multiple testing methods.
\end{abstract}

\section{Introduction} \label{sec:intro}

We consider the Empirical Bayes (EB) approach to the Normal Means problem \citep{efron1973,johnstone2004}:
\begin{equation}
X_j \mid \theta_j, s_j 
\sim N(\theta_j,s_j^2) \ , \qquad j=1, \dots, p \ .
\label{eq:nm}
\end{equation}
Here $N(\mu,\sigma^2)$ denotes the normal distribution with mean $\mu$
 and variance $\sigma^2$;
$X \defeq (X_1, \ldots, X_p)$ are observations; $s \defeq (s_1, \ldots, s_p)$ are standard deviations that are assumed known; and $\theta \defeq (\theta_1, \ldots, \theta_p)$ 
 are unknown
 means to be estimated. The EB approach
 assumes that $\theta_j$ are independent and identically distributed (iid) from
 some ``prior'' distribution,
\begin{equation}
  \theta_j \iid g(\cdot) \ , \qquad j=1, \dots, p \ ;
\label{eq:g}
\end{equation}
and performs inference for $\theta_j$ in two steps:
first obtain an estimate of $g$, $\hat{g}$ say, and second compute the posterior distributions $p(\theta_j|X_j, s_j, \hat{g})$.  We refer to the two-step process as ``solving the Empirical Bayes Normal Means (EBNM) problem.'' The first step, estimating $g$, is sometimes
of direct interest in itself, and is an example of a  ``deconvolution'' problem \citep[e.g.][]{kiefer1956,laird1978,kerneldeconv1990, fan1991,cordy1997,bovy2011,efron2016}.
 
 
First named by \cite{robbins1956}, 
EB methods have seen extensive theoretical study \citep[e.g.][]{robbins1964,morris1983,efron1996,jiang2009,brown2009,scott2010,petrone2014,rousseau2017,efron2018},
and are becoming widely used in practice.
Indeed, according to \cite{efron_casi}, ``large parallel data sets are a hallmark of twenty-first-century scientific investigation, promoting the popularity of empirical Bayes methods.'' 

The EB approach provides a particularly attractive solution to the Normal Means problem.
For example, the posterior means of $\theta$ provide shrinkage point estimates, with all the accompanying risk-reduction benefits
\citep{efron1972,berger1985}. And the posterior distributions for $\theta$ provide corresponding ``shrinkage'' interval estimates, which can have
good coverage properties even ``post-selection'' \citep{dawid1994,stephens2017}.
Further, by estimating $g$, EB methods ``borrow strength'' across observations, and automatically determine an appropriate amount of shrinkage from the data \citep{johnstone2004}. Because of these benefits, methods for solving the EBNM problem -- and related extensions -- are increasingly used in data applications \citep[e.g.][]{clyde2000,johnstone2005,brown2008,koenker2014,smash,mash,flash,corshrink}. One application of EBNM methods that we pay particular attention to later is large-scale multiple testing, and estimation/control of the False Discovery Rate \citep[FDR;][]{benjamini1995,efron_lsi, muralidharan2010,stephens2017,mouthwash}.


Almost all existing treatments of the EBNM problem assume that the observations $X$ in \eqref{eq:nm}
are independent given $\theta, s$.
However, this assumption can be grossly violated in practice. Non-negligible correlations are common in real world data sets. Further, as we discuss later, EB approaches to the Normal Means problem are particularly vulnerable to being misled by pervasive correlations. Specifically, when the average strength of pairwise correlations among observations is non-negligible, the estimate $\hat g$ of $g$ can be very inaccurate, and this adversely affects inference for {\it all} $\theta$. Ironically then, the attractive ``borrowing strength'' property of the EB approach becomes, in the presence of pervasive correlations, its Achilles' heel.

In this paper we introduce methods for solving the EBNM problem {\it allowing for unknown correlations} among the observations. 
More precisely, rewriting \eqref{eq:nm} as
\begin{align}
\begin{split}
X_j
&= \theta_j + s_jZ_j \\
Z_j &\sim N(0,1) \ ,
\end{split}
\label{eq:ebnm}
\end{align}
we develop methods that allow for unknown correlations among the ``noise'' $Z \defeq (Z_1, \ldots, Z_p)$.
Our methods are
built on elegant theory from \cite{schwartzman2010}, who shows, in essence, that the limiting empirical distribution, $f$ say, of correlated $N(0,1)$ random variables can be represented using a basis of the standard Gaussian density and its derivatives of increasing order. 
We use this result, combined with an assumption that $Z$ are exchangeable, to frame solving this ``EBNM with correlated noise'' problem
 as a two-step process: first 
 {\it jointly estimate $f$ and $g$} from all observations; and second compute the posterior distribution of $\theta_j$ given the estimated $\hat f,\hat g$ (and $X_j,s_j$). Although many possible assumptions on $g$ are possible, here we assume $g$ to be a scale mixture of zero-mean Gaussians, following the
flexible ``adaptive shrinkage'' approach in \cite{stephens2017}. We have implemented these methods in an \R{} package, \cash{} (``correlated 
adaptive shrinkage in \R{}''), available from \url{https://github.com/LSun/cashr}. 


The rest of the paper is organized as follows.
In Section \ref{sec:motivation}, we illustrate how correlation can derail existing EBNM methods, and review \cite{schwartzman2010}'s representation of the empirical distribution of correlated $N(0,1)$ random variables. In Section \ref{sec:cash} we introduce the exchangeable correlated noise (ECN) model, and describe methods to solve the EBNM with correlated noise problem. Section \ref{sec:examples} provides numeric examples with realistic simulations and real data illustrations. Section \ref{sec:disc} concludes and discusses future research directions.

\section{Motivation and Background} \label{sec:motivation}

\subsection{Correlation distorts empirical distribution and misleads EBNM methods} \label{sec:distortion}

In essence, the reason correlation
can cause problems for EBNM methods is
that, even with large samples, the empirical distribution of correlated variables can be quite different from their marginal distribution \cite[e.g.][]{efron2007jasa}.
To illustrate this,
we generated realistic correlated $N(0, 1)$ $z$-scores using a framework
similar to 
\cite{gerard.ruv,mouthwash,mengyin_thesis}.
Specifically, we took RNA-seq
gene expression data on the $10^4$ most highly expressed genes in 119 human liver tissues \citep{gtex2015,gtex2017}. In each simulation we randomly drew two groups of five samples (without replacement), and applied a standard RNA-seq analysis pipeline, using the software packages \edgeR{} \citep{edgeR}, \voom{} \citep{voom}, and \limma{} \citep{limma}, to compute, for each gene $j = 1, \ldots, 10^4$, an estimate of the $\log_2$-fold difference in mean expression, $X_j$, and a corresponding $p$-value, $p_j$, testing the null hypothesis that the difference in mean is 0. We converted $p_j,X_j$ to a $z$-score $z_j := -\text{sign}(X_j)\Phi^{-1}(p / 2)$, where $\Phi$ is the CDF of $N(0, 1)$. We also computed an ``effective'' standard deviation $s_j \defeq X_j / z_j$ for later use (Figure \ref{fig:deconv} and Section \ref{sec:examples}).


In these simulations, because samples are randomly assigned to the two groups, there are no genuine differences in mean expression. Therefore the $z$-scores 
should have marginal distribution $N(0,1)$. And, indeed, empirical checks confirm that the analysis pipeline produces well-calibrated marginally $N(0,1)$ $z$-scores (Appendix \ref{sec:marginal_N01}).
However, the $10^4$ $z$ scores in each simulated data set are correlated, due to correlations among genes, 
and such correlations can distort the 
empirical distribution so that it is very different from $N(0,1)$ \citep{efron2007jasa,efron2010,efron_lsi}.
Figure \ref{fig:cor_z} shows four examples, which were chosen to highlight some common patterns. 
 Panels (a-c) all exhibit a feature we call {\it pseudo-inflation}, where
 the empirical distribution is {\it more} dispersed than $N(0,1)$. Conversely, 
 panel (d) exhibits {\it pseudo-deflation}, where the
 empirical distribution is {\it less} dispersed than $N(0,1)$. Panel (b) also exhibits skew.

\begin{figure*}[!htb]
\begin{center}
\includegraphics[width = 0.9\linewidth]{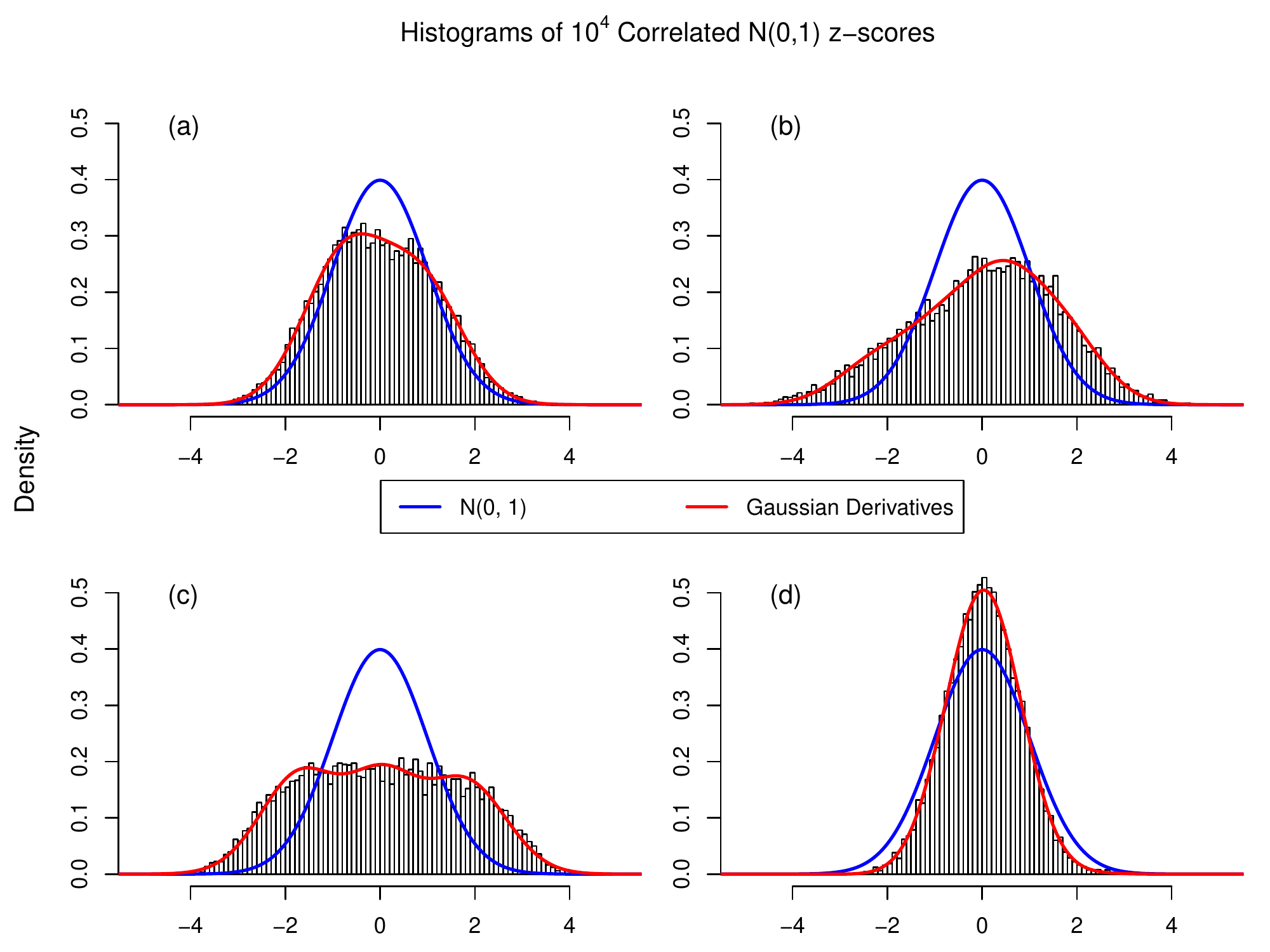}
\caption{Illustration that the empirical distribution of a large number of correlated and marginally $N(0,1)$ null $z$-scores can deviate substantially from $N(0, 1)$. The red lines are fitted densities obtained using our ``Exchangeable Correlated Noise" model (Section \ref{fitting_z}) which uses a linear combination of the standard Gaussian density and its standardized derivatives.}
\label{fig:cor_z}
\end{center}
\end{figure*}

 Such {\it correlation-induced distortion} of the empirical distribution, if not appropriately addressed, can have serious consequence for EBNM methods. To illustrate this we applied several EBNM methods to five data sets simulated
according to \eqref{eq:g}-\eqref{eq:ebnm} as follows:
\begin{itemize}
\item The $p = 10^4$ normal means $\theta$ are iid samples from the mixture
$g(\cdot) = 0.6\delta_0(\cdot) + 0.3N(\cdot;0, 1) + 0.1N(\cdot; 0, 3^2)$,
where $\delta_0(\cdot)$ denotes a point mass on $0$ whose coefficient ($0.6$) is the null proportion, and $N(\cdot ; \mu, \sigma^2)$ denotes the Gaussian density with mean $\mu$ and variance $\sigma^2$. The same $\theta$ are used in all five data sets.
\item In the first four data sets, the noise variables, $Z$, are the correlated null $z$-scores from the four panels of Figure \ref{fig:cor_z}. In the fifth data set $Z$ are iid $N(0,1)$ samples.
\item The standard deviations $s$ are simulated from the RNA-seq gene expression data as described above, and $s$ are scaled to have $\frac1{p}\sum_j s_j^2 = 1$.
\end{itemize}
We provide the simulated $X,s$ values to four
existing EBNM methods -- \EbayesThresh{} \citep{johnstone2004,EbayesThresh}, \REBayes{} \citep{koenker2014,koenker2017}, \ash{} \citep{stephens2017}, and \deconvolveR{} \citep{efron2016,deconvolveR} --
that all ignore correlation and assume independence among observations. (For \deconvolveR{} we set $s_j \equiv 1$ as its current implementation supports only homoskedastic noise.)

The estimates of $g$ obtained by each method are shown in Figure \ref{fig:deconv}.
All methods do reasonably well in the fifth data set where $Z$ are indeed independent (panel (e)). However, in the correlated data sets (panels (a-d)) the methods all misbehave in a similar way: over-estimating the dispersion of $g$ under pseudo-inflation, and under-estimating it under pseudo-deflation. Their estimates of the null proportion are particularly inaccurate. These adverse effects are visible even when the distortion appears not severe (e.g. Figure \ref{fig:cor_z}(a)). 

\begin{figure*}[!htb]
\begin{center}
\includegraphics[width = \linewidth]{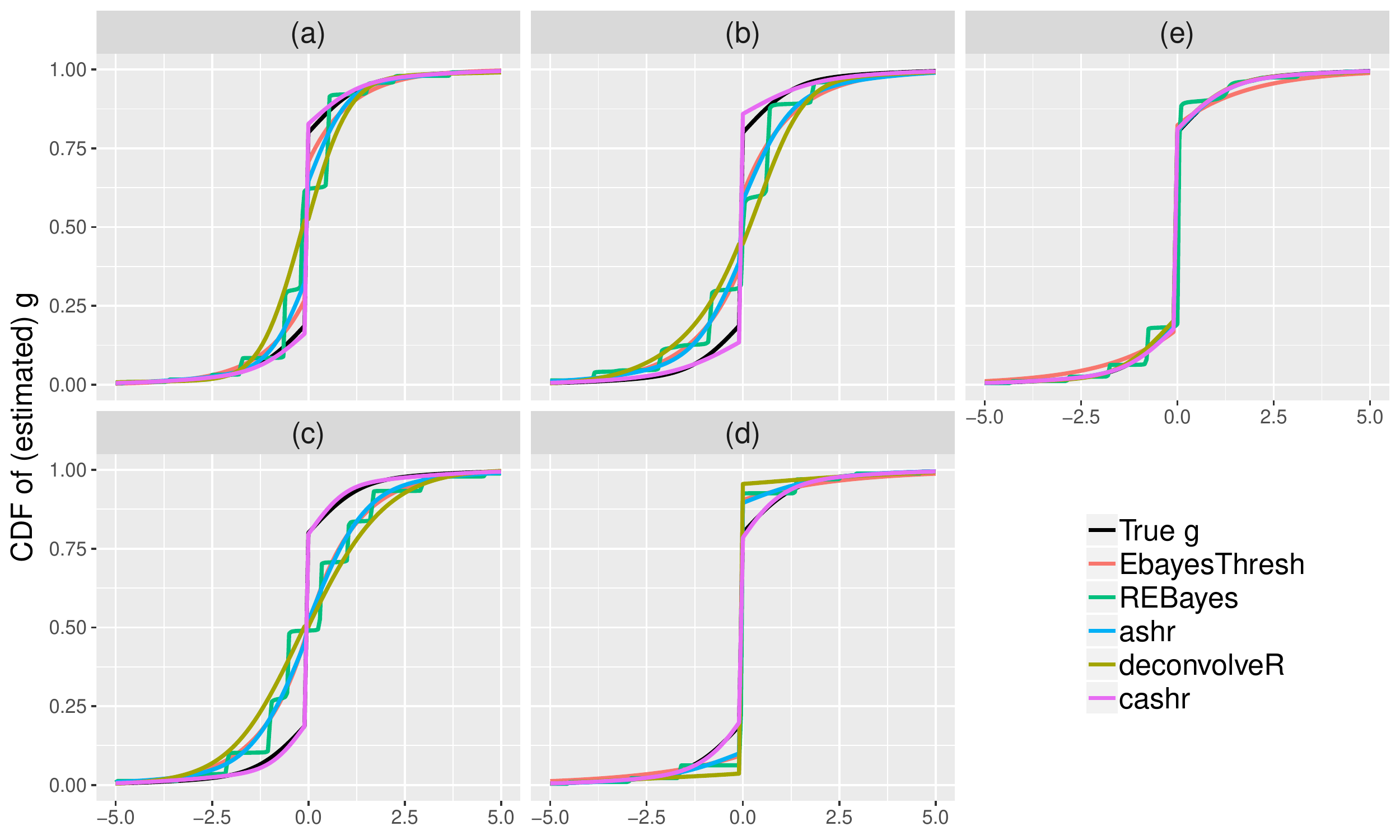}
\caption{Illustration of how correlation can distort estimates of $g$ obtained by EBNM methods. Each panel compares the true $g$ with the estimated $g$ from several EBNM methods applied to the same simulated dataset (see main text for details). In panels (a-d) $Z$ are the correlated null $z$-scores from the corresponding panels of Figure \ref{fig:cor_z}. In panel (e) $Z$ are iid $N(0,1)$ samples.  Existing EBNM methods (\EbayesThresh{}, \REBayes{}, \ash{}, \deconvolveR{}), which ignore correlation among observations, do reasonably well with iid noise (e). However they do much less well in the correlated cases (a-d): over-estimating the dispersion of $g$ under pseudo-inflation (a-c) and under-estimating it under pseudo-deflation (d).   In contrast, our new method \cash{} (Section \ref{sec:cash}) estimates $g$ consistently well.}
\label{fig:deconv}
\end{center}
\end{figure*}

As a taster for what is to come, Figure \ref{fig:deconv} also shows the results from our new method, \cash{}, described later. This new method can account for both pseudo-inflation and pseudo-deflation, and in these examples estimates $g$ consistently well.

\subsection{Pseudo-inflation is non-Gaussian} \label{sec:nongaussian}

In a series of pioneering papers \citep{efron2004,efron2007jasa,efron2007aos,efron2008,efron2010}, Efron studied the impact of correlations among $z$-scores on EB approaches to multiple testing. 
To account for
the effects of correlation (pseudo-inflation, pseudo-deflation, and skew) on the empirical distribution of null $z$-scores he introduced the concept of an ``empirical null.'' In his \locfdr{} method \citep{locfdr}, the empirical null is assumed to be Gaussian $N(\mu_0, \sigma_0^2)$. However, theory suggests that pseudo-inflation is not well modelled by a Gaussian distribution \citep[][reviewed in Section \ref{sec:cor_z}]{schwartzman2010},
and a closer look at our
empirical results here supports this conclusion.

To illustrate, Figure \ref{fig:shoulder} shows more detailed analysis of the empirical distribution of Figure \ref{fig:cor_z}(c) $z$-scores. The central part of this $z$-score
distribution could perhaps be modelled by
a Gaussian distribution with inflated variance -- for example, it matches more closely to a $N(0,1.6^2)$ than to $N(0,1)$. However,
in the tails, the empirical distribution has much shorter tails than $N(0,1.6^2)$. For example, $10^4$ iid $N(0,1.6^2)$ samples would be 
expected to yield 43 $p$-values exceeding the Bonferroni threshold of $0.05/10^4$, whereas in fact we observe none here. In short, the effects of pseudo-inflation are primarily in the ``shoulders'' of the distribution, where $|z|$-scores are only moderately large, and not in the tails. (Incidentally, this behavior is far more evident in the histogram of $z$-scores than in the histogram of corresponding $p$-values, and we find the $z$-score histogram generally more helpful for diagnosing potential correlation-induced distortion.)



\begin{figure*}[!htb]
\begin{center}
\includegraphics[width = 0.375\linewidth]{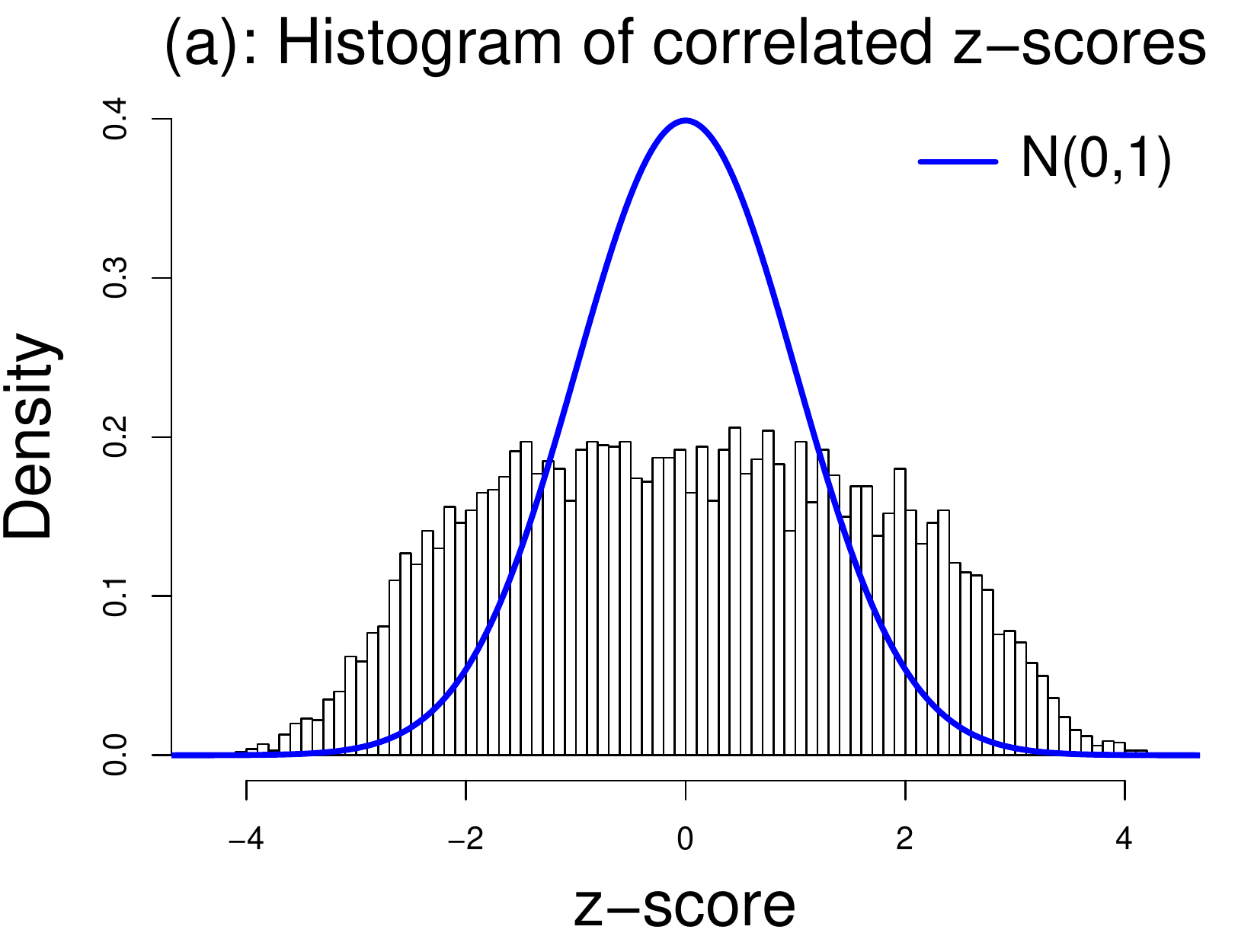}
\includegraphics[width = 0.375\linewidth]{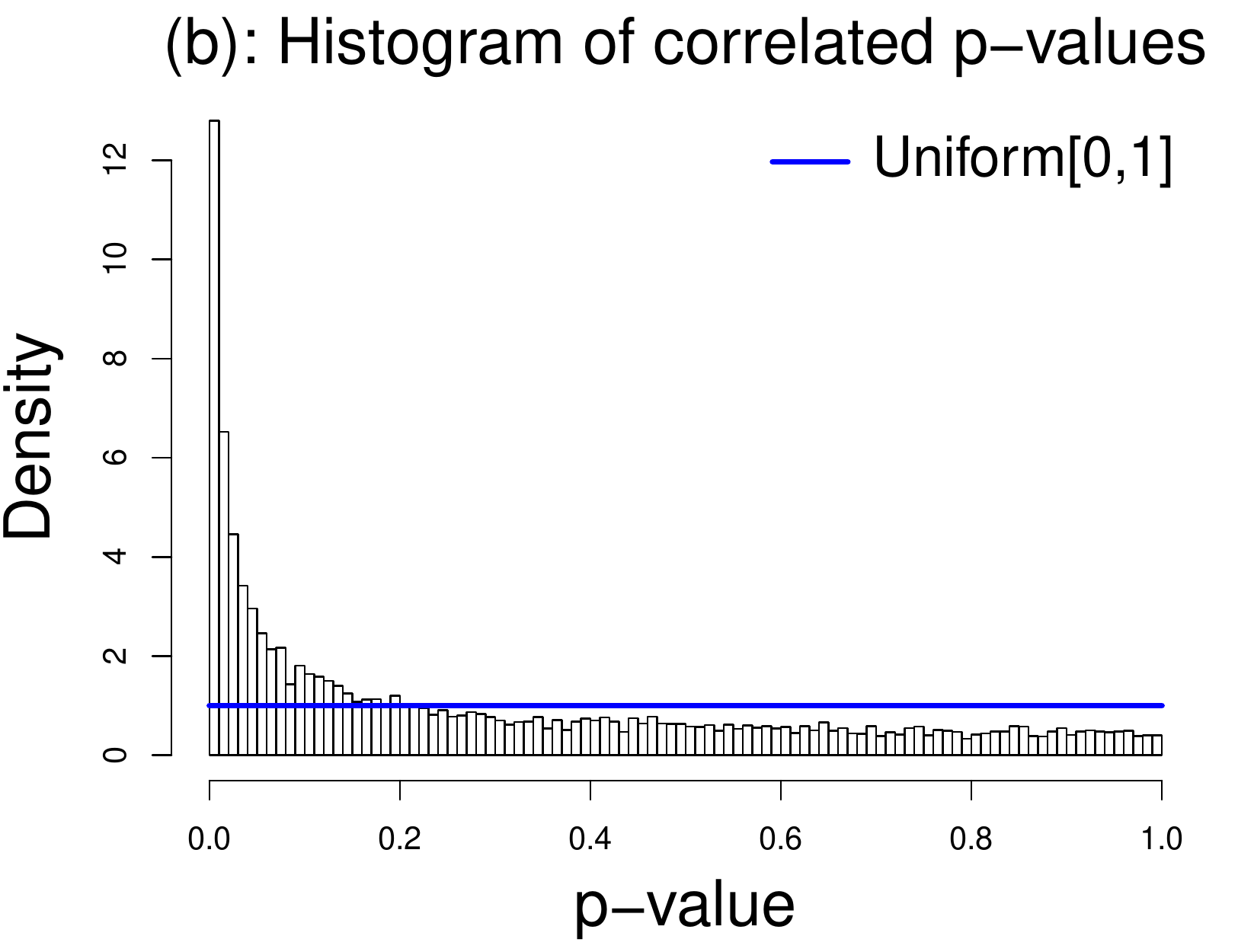} \\
\vspace{0.5cm}
\includegraphics[width = 0.3\linewidth]{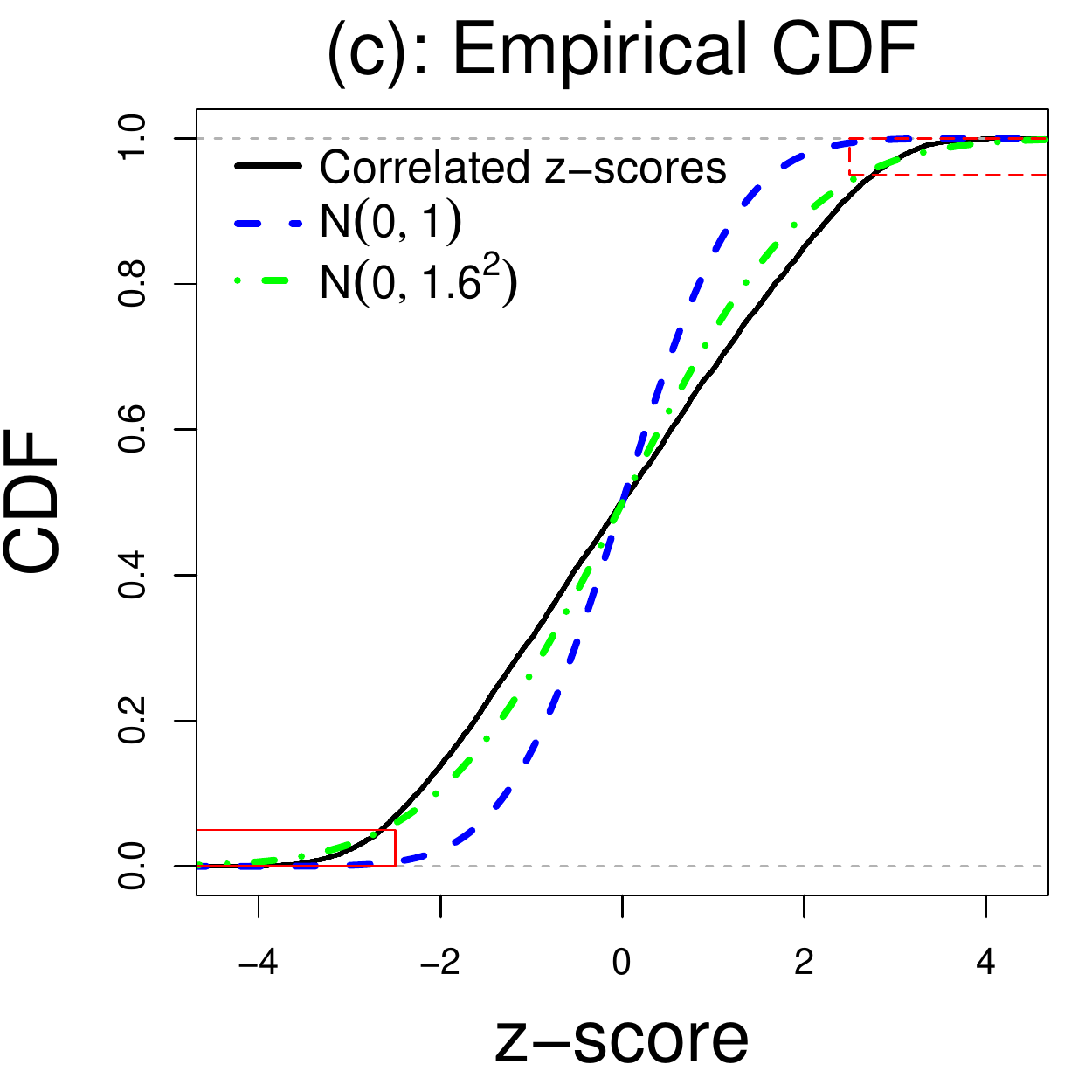}
\includegraphics[width = 0.3\linewidth]{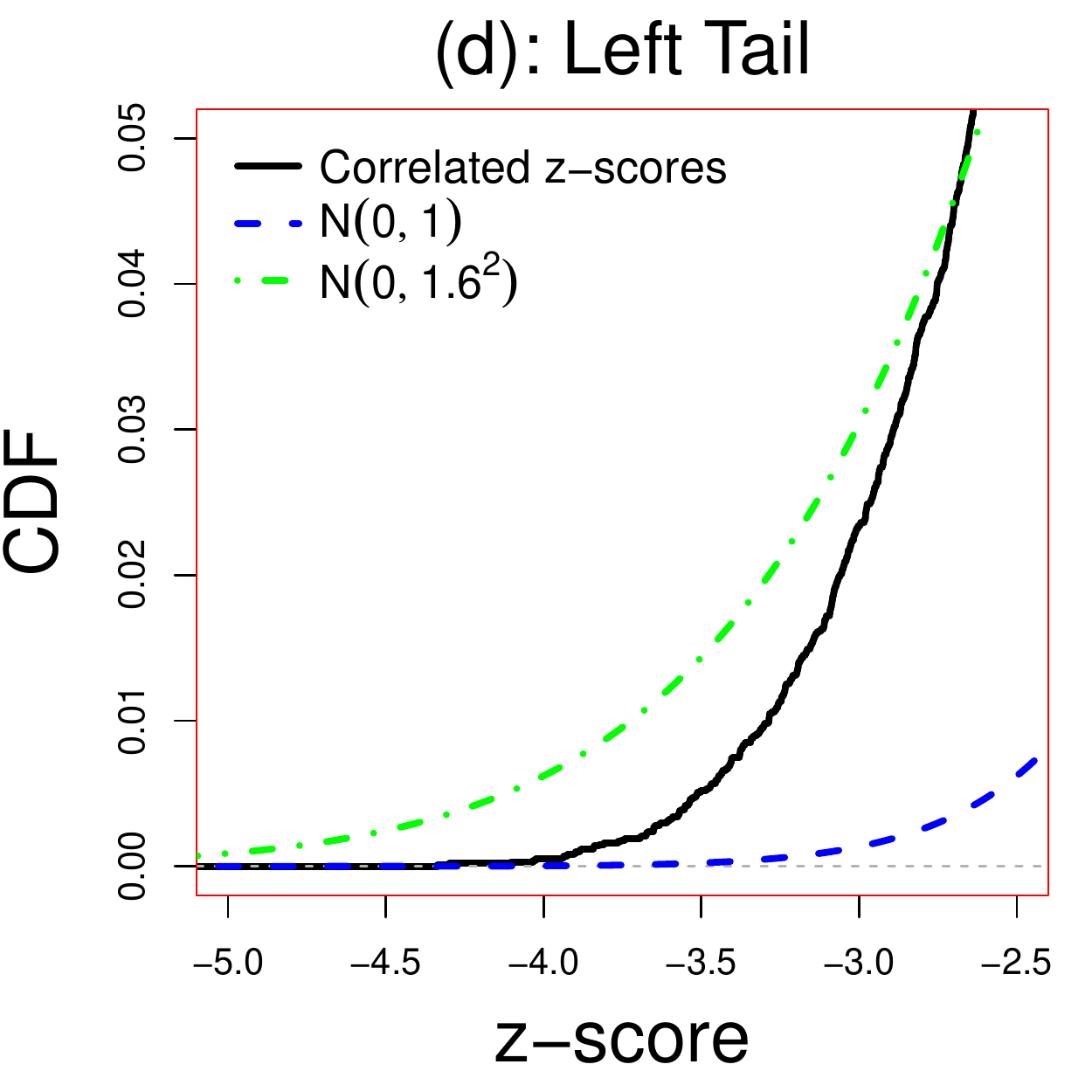}
\includegraphics[width = 0.3\linewidth]{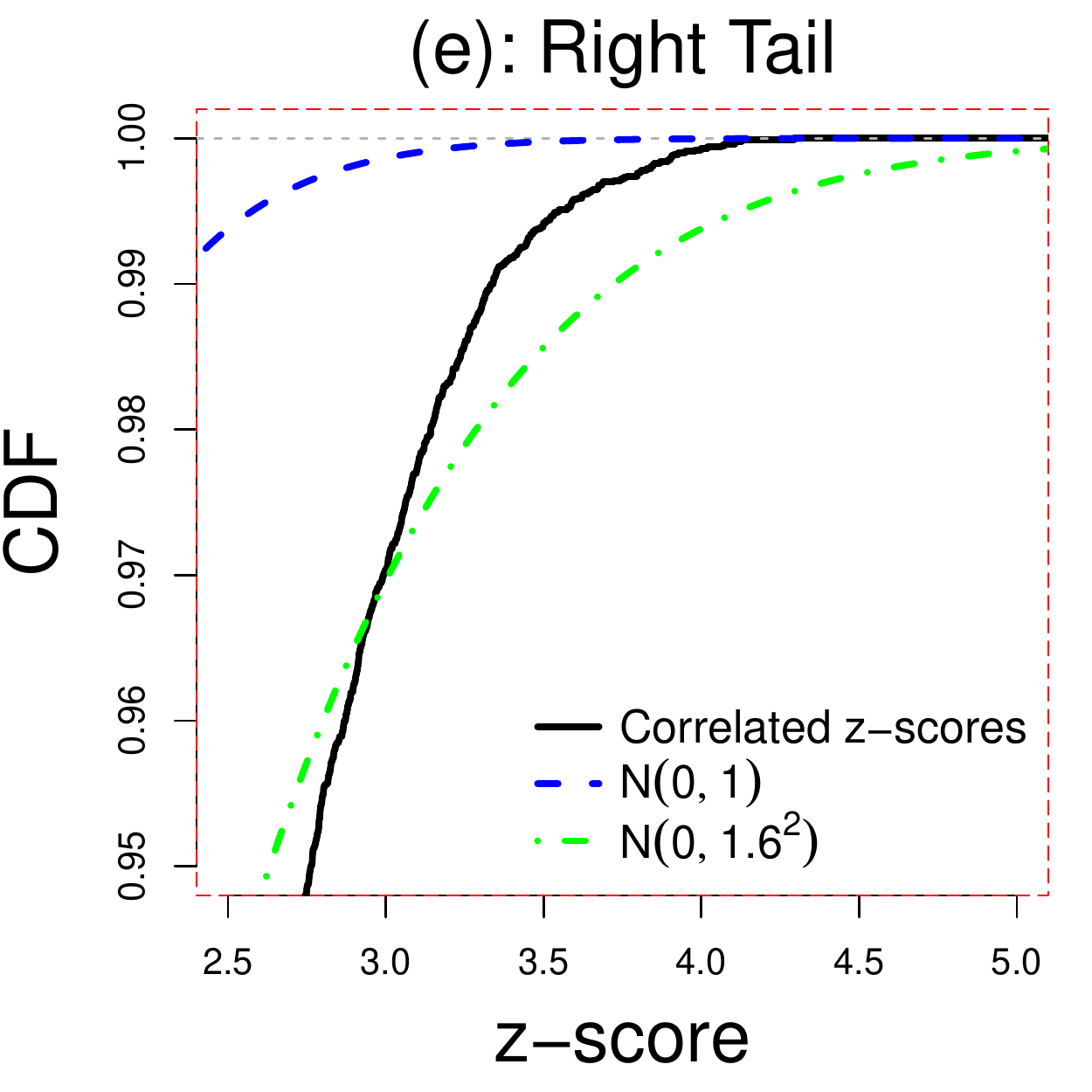} \\ \vspace{0.5cm}
\includegraphics[width = 0.3\linewidth]{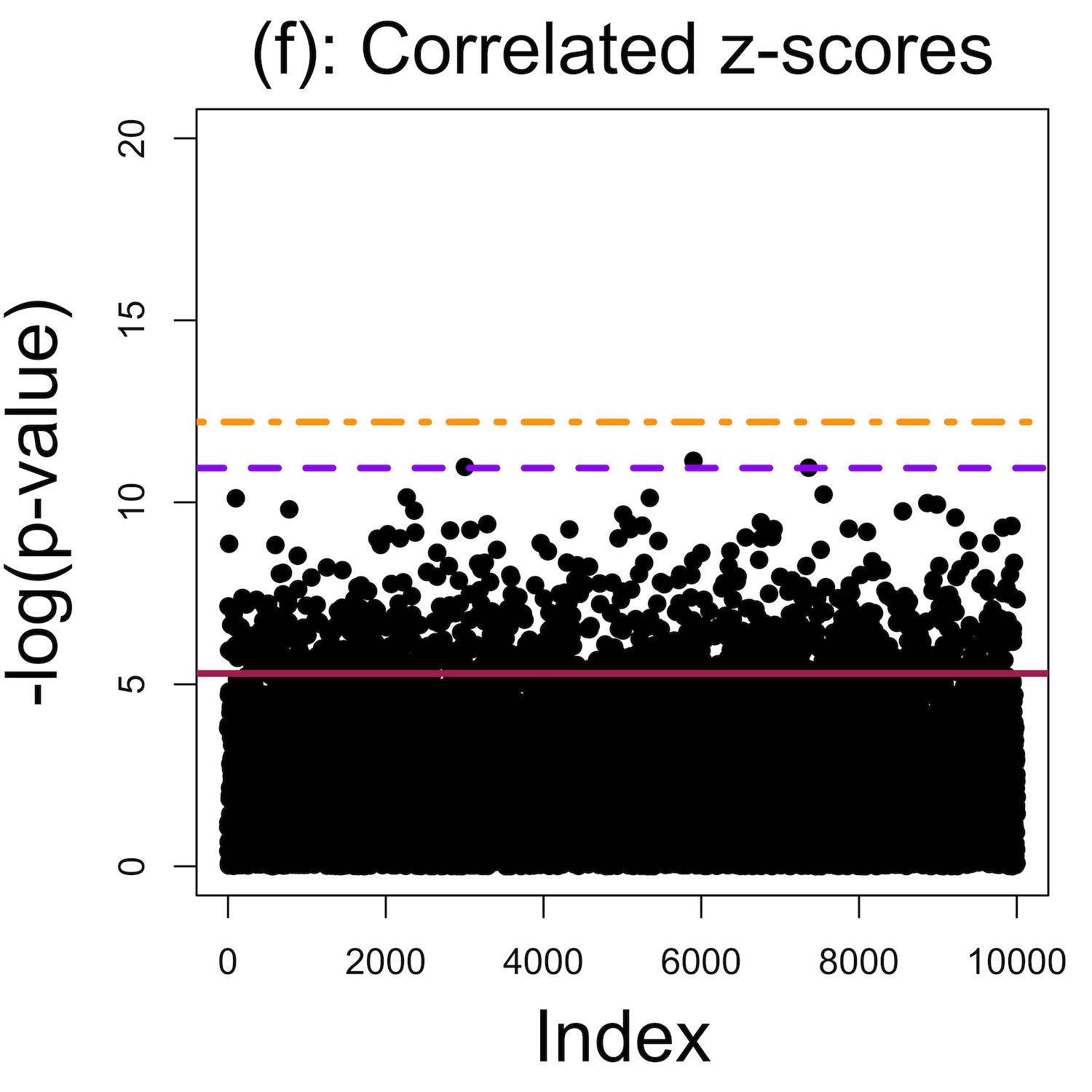}
\includegraphics[width = 0.3\linewidth]{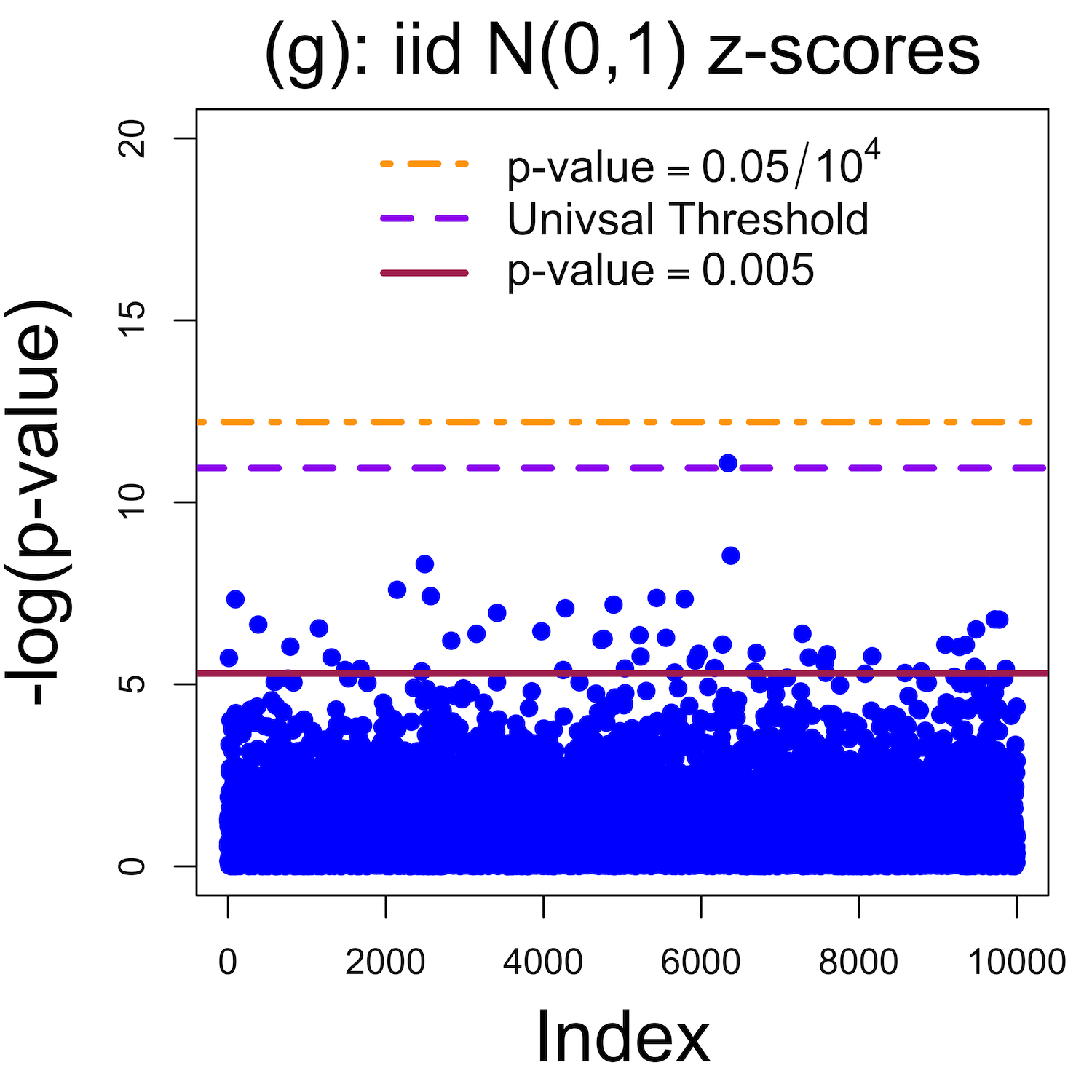}
\includegraphics[width = 0.3\linewidth]{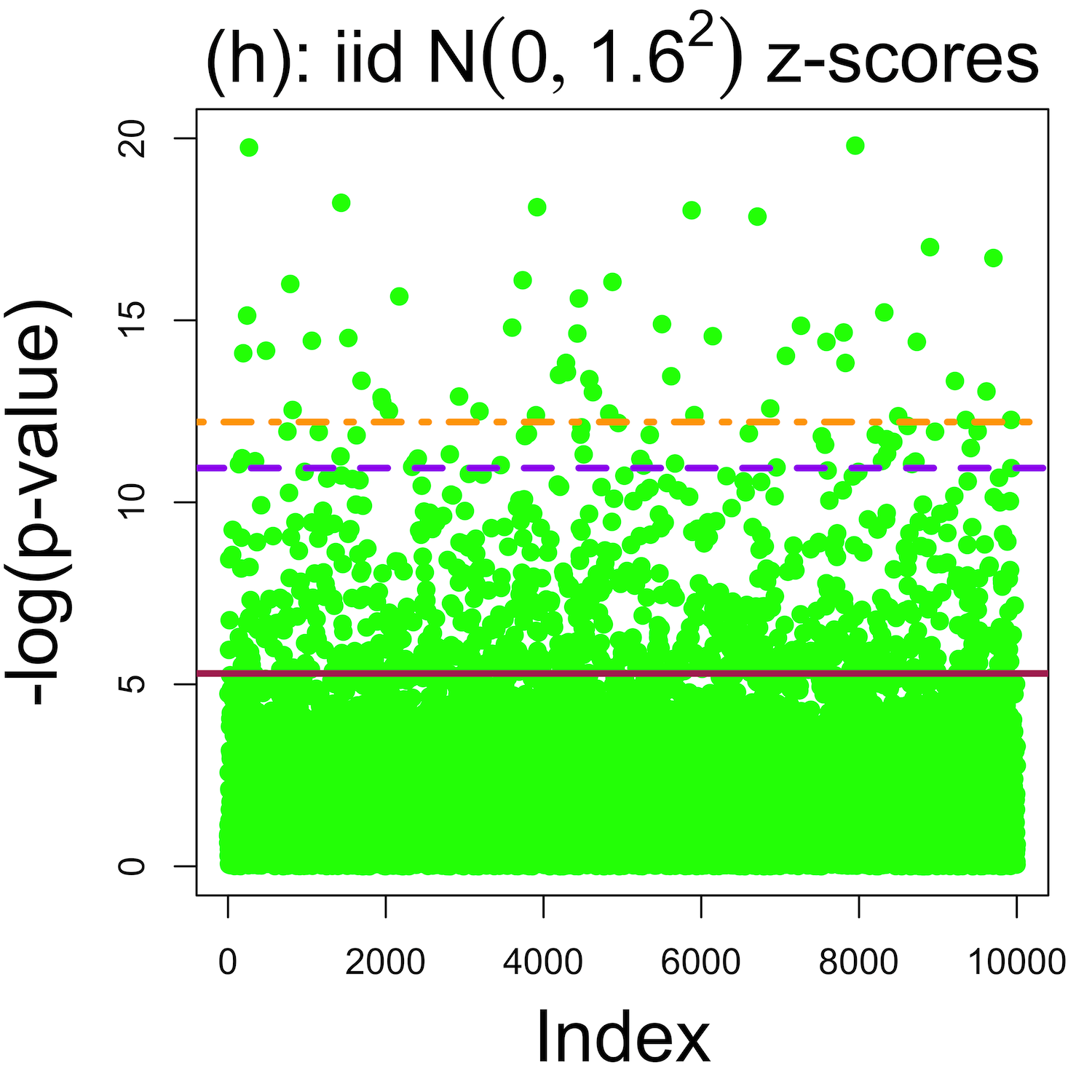} \\ \caption{Illustration that the effects of pseudo-inflation are primarily in the ``shoulders'' of the distribution of null $z$-scores, and not in the tails.
Panels (a-b): Histograms of correlated $z$-scores (from Figure \ref{fig:cor_z}(c)) and their corresponding $p$-values. Note that the "shoulder-but-not-tail" inflation is evident in the histogram of $z$-scores (a) but not in the oft-used histogram of $p$-values (b).
Panels (c-e): Comparison of the empirical CDF of correlated $z$-scores with the CDF of $N(0, 1)$ and $N(0, 1.6^2)$. 
The $z$-score distribution is closer to $N(0,1.6^2)$ in the center, but closer to $N(0,1)$ in the tails.
Panels (f-h): Comparison of correlated $p$-values   
with $p$-values obtained from $10^4$ iid  $N(0, 1)$ and $N(0, 1.6^2)$ $z$-scores.
The number of correlated $p$-values $\le 0.005$ is closer to $z$ scores from
$N(0,1.6^2)$, but the number in the extreme tail (e.g. clearing the Bonferroni or universal thresholds) is closer to $N(0,1)$.
}
\label{fig:shoulder}
\end{center}
\end{figure*}

With hindsight this shoulder-but-not-tail inflation pattern should perhaps be expected.
If one views the effect of correlation as to reduce the effective sample size, the number of extreme values of a sample with a smaller effective sample size should indeed be smaller. There are also relevant discussions on ``asymptotic independence'' in the extreme value theory \citep{sibuya1960,carvalho2012}.  However, this property of pseudo-inflation does suggest that using a Gaussian to describe correlation-induced distortion,  
as in \locfdr{}, is not ideal (more discussion in Section \ref{sec:examples}).

\subsection{Empirical distribution of correlated $N(0,1)$ random variables} \label{sec:cor_z}

We now summarize an elegant result of \cite{schwartzman2010}, which characterizes the empirical distribution of a large number of correlated $N(0, 1)$ $z$-scores. This result plays a key role in our work.

On notation: let $\varphi$ denote the PDF of $N(0,1)$, and $\varphi^{(l)}$ denote the $l^\text{th}$ derivative of $\varphi$. We refer to the collection of functions $\left\{\frac{1}{\sqrt{l!}}\varphi^{(l)}\right\}_{l = 1}^\infty$ as the (standardized) Gaussian derivatives. (Here ``standardized" means that they are scaled to be orthonormal with respect to the weight function $\varphi$.) 

Let $Z \defeq \{Z_1, \ldots, Z_p\}$ be $p$ identically distributed, but not necessarily independent, $N(0, 1)$ random variables.  Let $F_p$ denote their empirical CDF:
\begin{equation}
F_p(\cdot)\defeq \frac1p\sum\limits_{j = 1}^p \mathcal{I}(Z_j \leq \cdot) \ ,
\end{equation}
where the indicator function $\mathcal{I}(Z_j \leq \cdot) \defeq \begin{cases}1 & Z_j \leq \cdot\\ 0 & Z_j > \cdot\end{cases}$. 
Since $Z$ are random variables, $F_p$ is a random function on $\mathbb{R} \to \left[0, 1\right]$. Also, because $Z$ are marginally $N(0,1)$, the expectation of $F_p$ is $\Phi$. 

\cite{schwartzman2010} studies the distribution of $F_p$, and how its deviation from the expectation $\Phi$
depends on the correlations among $Z$.
Specifically, assuming that each pair $\{Z_i,Z_j\}$ is bivariate normal with correlation $\rho_{ij}$ (which is weaker than the common assumption that all $Z$ are joint multivariate normal), \cite{schwartzman2010} provides the following
representation of $F_p$ when $p$ is large:
\begin{equation} \label{eq:Fbasis}
F_p(\cdot) \approx F(\cdot):=\Phi(\cdot) + \sum_{l = 1}^\infty W_l \frac{1}{\sqrt{l!}} \varphi^{(l - 1)}(\cdot) \ ,
\end{equation}
where $W_1, W_2, \ldots$ are uncorrelated random variables with $E[W_l] = 0$, and
\begin{equation}
var(W_l) = \overline{\rho^l} \defeq \frac{1}{p(p - 1)}\sum\limits_{i,j: i \neq j}\rho_{ij}^l \ .
\end{equation}
Although uncorrelated, $W_1, W_2, \ldots$ are not independent; they must have higher-order dependence to guarantee that $F$ is non-decreasing. Also here we assume $\overline{\rho^l} \geq 0$ for all $l \in \mathbb{N}$. Note that this assumption should not be too demanding for large $p$ in practice \citep[][also see Appendix \ref{sec:schwartzman2010}]{schwartzman2010}.

Since $F$ is a CDF, its derivative defines a corresponding density:
\begin{equation} \label{eq:gd_z_1}
f(\cdot) \defeq F'(\cdot) = \varphi(\cdot) + \displaystyle\sum\limits_{l = 1}^\infty W_l \frac{1}{\sqrt{l!}}\varphi^{(l)}(\cdot) \ .
\end{equation}
Intuitively, \eqref{eq:gd_z_1} characterizes how
the (limiting) empirical distribution (histogram) of $Z$ is likely to randomly deviate from the expectation $\varphi$, using standardized Gaussian derivatives as basis functions.

The representation \eqref{eq:gd_z_1} is crucial to our work here, and
provides some remarkable insights. We highlight particularly the following:
\begin{enumerate}
\item \label{rhobar} The expected deviations of $f$ from $\varphi$ are determined by the variances of the coefficients $W_l$, which are determined by the mean and higher moments of the pairwise correlations, $\overline{\rho^l}$. In the special case where $Z$ are independent all these terms are $0$, and $f=\varphi$.
\item Following from \ref{rhobar}, to create a tangible deviation from $\varphi$, $\overline{\rho^l}$ must be non-negligible (for some $l$). This requires {\it pervasive, but not necessarily strong,} pairwise correlations. For example, pervasive correlations occur if there is an underlying low-rank structure in the data, where all $Z$ are influenced by a small number of underlying random factors, and so are all correlated with one another. In this case $\overline{\rho^l}$ will be non-negligible, and $f$ may deviate from $\varphi$.
In contrast, there can exist very strong pairwise correlations with negligible effect on $f$. For example, suppose $p$ is even, and let $Z$ be in $p / 2$ pairs, with each pair having correlation one but different pairs being independent. The histogram of $Z$ will look very much like $N(0,1)$, because $\overline{\rho^l} = \frac{1}{p - 1} \approx 0$ for large $p$. In other words, not all kinds of correlations, even large ones, distort the empirical distribution of $Z$.
\item Barring special cases such as $\rho_{ij}=1$, the moments $\overline{\rho^l}$, and hence
the expected magnitude of $W_l$, will decay quickly with $l$. Consequently the sum in \eqref{eq:gd_z_1} will typically be dominated by the first few terms,
and the shape of the first few basis functions will determine the typical deviation of $f$ from $\varphi$. Of the first four basis functions (Figure \ref{fig:GD}), the \nth{2} and \nth{4} correspond to pseudo-inflation or pseudo-deflation in the shoulders of $\varphi$, depending on the signs of their coefficients, whereas the \nth{1} and \nth{3} correspond to mean shift and skewness.
This explains the empirical observation that correlation-induced pseudo-inflation tends to focus in the shoulders, and not the tails. (Also see Appendix \ref{sec:gauss_sgd} for the special case $\rho_{ij} = 1$.)
\end{enumerate}

\begin{figure*}[!htb]
\begin{center}
\includegraphics[width = \linewidth]{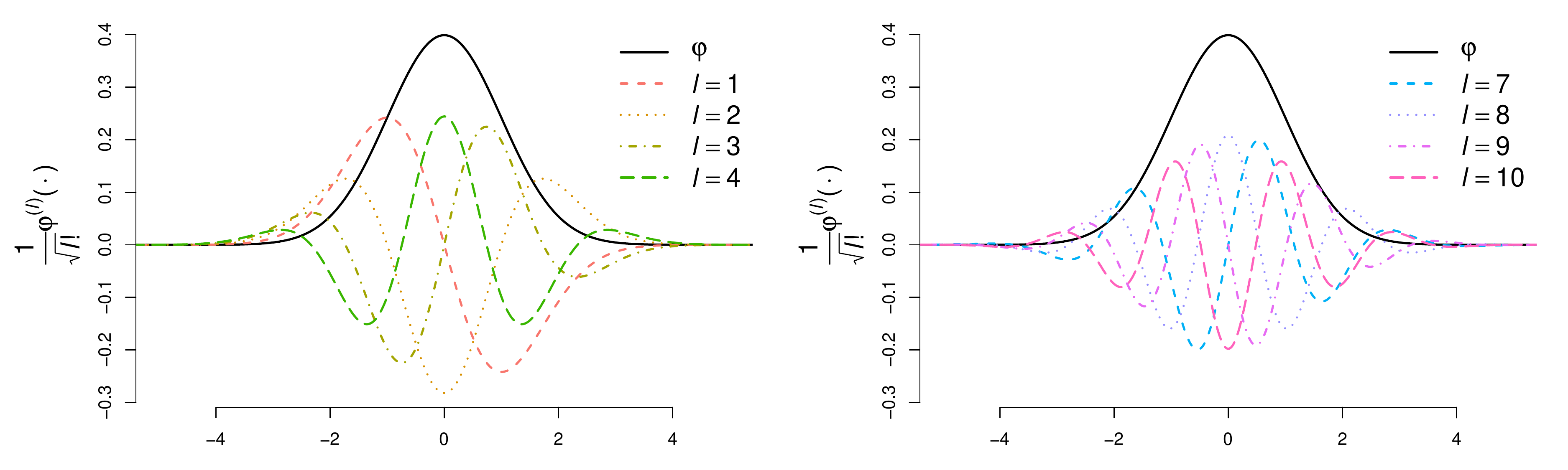}
\caption{Illustration of the standard Gaussian density, $\varphi$, and its standardized derivatives. The left panel shows $\varphi$ and its first four standardized derivatives. The \nth{2} and \nth{4} derivatives correspond to pseudo-inflation or pseudo-deflation in the shoulders; the \nth{1} and \nth{3} derivatives correspond to mean shift or skewness. The right panel shows $\varphi$ and its \nth{7}-\nth{10} derivatives. Even for these higher-order derivatives, tails are short, implying that correlation-induced distortion is unlikely to have long tails.}
\label{fig:GD}
\end{center}
\end{figure*}

In discussing \cite{efron2010}, \cite{schwartzman2010} used this result to argue that ``a wide unimodal histogram (of $z$-scores) may be indication of the presence of true signal, rather than an artifact of correlation.'' Specifically, by discarding terms for $l \geq 4$ in \eqref{eq:gd_z_1}, he found that the largest central spread (standard deviation) for $f$ in \eqref{eq:gd_z_1} being a unimodal density is approximately 1.3. Along similar lines, we can show (Appendix \ref{sec:gauss_sgd}) that the maximum standard deviation for $f$ being a Gaussian density is $\sqrt{2} \approx 1.4$. The key point here is that the effects of correlation are different from the effects of true signals, so the two can (often) be separated. Our methods here are designed to do exactly that.


\section{Empirical Bayes Normal Means with Correlated Noise} \label{sec:cash}

\subsection{The Exchangeable Correlated Noise model} \label{fitting_z}

As a first step towards allowing for correlated noise in the EBNM problem, we develop
methods to fit the representation \eqref{eq:gd_z_1} to correlated null $z$-scores.
We do this by treating $Z$ as conditionally iid samples from $f$ in \eqref{eq:gd_z_1}, parameterized by $\omega \defeq \{\omega_1, \omega_2, \ldots \}$ which are realizations of $W \defeq \{W_1, W_2, \ldots\}$:
\begin{equation}
Z_j \mid \{W = \omega\} \iid f(\cdot; \omega) := \varphi(\cdot) + \sum\limits_{l = 1}^\infty
\omega_l \frac{1}{\sqrt{l!}}\varphi^{(l)}(\cdot)
\ .
\label{eq:f_omega}
\end{equation}
It may seem perverse to model correlated random variables as conditionally iid. However, this treatment can be motivated by assuming $Z$ are exchangeable and appealing to de Finetti's representation theorem \citep{definetti1937}, which says that (infinitely) exchangeable random variables can be represented as being conditionally iid from their empirical distribution. We therefore refer to the model \eqref{eq:f_omega} as the {\it exchangeable correlated noise (ECN)} model. We also refer to $f$ as the {\it correlated noise distribution}.

To fit the ECN model \eqref{eq:f_omega} with observed $Z$, we estimate $\omega$, essentially by maximum likelihood, but
with a couple of additional complications that we now describe. 
First, since $f$ is a density, we must constrain the parameters $\omega$ to ensure that $f(\cdot; \omega)$ is non-negative (note that \eqref{eq:f_omega} integrates to one for any $\omega$, but is not guaranteed to be non-negative). Ideally $f$ should be non-negative on the whole real line, but this constraint is difficult to work with, so we approximate it using a discrete approximation: we constrain $f(\mathfrak{z}_i; \omega) \ge 0$ on a fine grid $\{\mathfrak z_1, \ldots, \mathfrak z_m\}$ such as $\{-10, -9.999, -9.998, \ldots, +9.998, +9.999, +10\}$, in addition to $f(Z_j; \omega)\ge 0$ for all $j$.

Second, to incorporate the prior expectation that the
absolute value of $\omega_l$ should decay quickly with $l$ 
(because $var(W_l) = \overline{\rho^l}$)
we introduce a penalty on $\omega$, 
\begin{equation}
h(\omega) \defeq \sum_l\gamma_l |\omega_l|,
\label{eq:penalty_h}
\end{equation}
where we take the penalty parameters $\gamma_l$ to be
\begin{equation}
\gamma_l =
\begin{cases}
0 & l \text{ is odd} \\
\gamma / \rho^{l/2} & l \text{ is even}
\end{cases} \ ,
\label{eq:gamma}
\end{equation}
where $\gamma$ represents a common penalty, and $\rho$ some notion of average pairwise correlation. For computational convenience we use only the first $L=10$ Gaussian derivatives (see Figure \ref{fig:GD} for \nth{7}-\nth{10} standardized Gaussian derivatives) and set $\omega_l=0$ for $l>10$. 
(Recall that $var(W_l) = \overline{\rho^l}$, so $W_l$'s realization $\omega_{l}$ will generally be negligible in practice for $l > 10$.)
Of course a full Bayesian treatment would attempt to account for uncertainty in $\omega$; in ignoring that here we are making the usual EB compromise.

In numerical simulations we experimented with different combinations of $\gamma \in \{1, 5, 10, 50, 100\}$ and $\rho \in \{0.10, 0.25, 0.50, 0.75, 0.90\}$, and found that $\gamma = 10, \rho = 0.5$ performed well in a variety of situations, although results were not very sensitive to the choice of $\gamma$ and $\rho$. All results in this paper were obtained with $\gamma = 10, \rho = 0.5$.



In summary, we estimate $\omega$ by solving:
\begin{equation}
\begin{array}{rll}
\max\limits_{\omega} & \sum_j \log f(Z_j; \omega) - h(\omega)  \\
\text{s.t.} 
& f(Z_j;\omega) \geq 0, & j = 1, \ldots, p \ ,
\\
& f(\mathfrak z_i;\omega) \geq 0, & i = 1, \ldots, m \ .
\end{array}
\label{prob:z_optim_simple}
\end{equation}
This is a convex optimization and can be solved efficiently and stably
using an interior point method; we
implemented this using the \R{} package \Rmosek{} to interface
to the MOSEK commercial solver \citep{rmosek}. With $p = 10^4$, the problem is solved on average within 0.50 seconds on a personal computer (Apple iMac, 3.2 GHz, Intel Core i5).

Figure \ref{fig:cor_z} shows the fitted distributions from the ECN model, $\hat f(\cdot;\hat\omega) := \varphi(\cdot) + \sum\limits_{l = 1}^{L}
\hat\omega_l \frac{1}{\sqrt{l!}}\varphi^{(l)}(\cdot)$, on the four illustrative sets of correlated null $z$-scores.




\subsection{The EBNM model with correlated noise} \label{sec:cash_model}


To allow for correlated noise in the EBNM problem, we combine the standard EBNM model \eqref{eq:g}-\eqref{eq:ebnm} with the ECN model \eqref{eq:f_omega}:
\def\G{\mathcal{G}}
\begin{align} \label{eq:xcond}
X_j & = \theta_j + s_j Z_j \\
\label{eq:g_prior}
\theta_j & \sim g(\cdot) \\ 
\label{eq:ecsn_z}
Z_j \mid \omega & \sim f(\cdot; \omega) = \varphi(\cdot) + \sum\limits_{l = 1}^L \omega_l\frac{1}{\sqrt{l!}}\varphi^{(l)}(\cdot) \ .
\end{align}
Note that in this model the observations are conditionally independent given $f$ and $g$.


Following \cite{stephens2017} we model the prior distribution $g$ by a finite mixture of zero-mean Gaussians:
\begin{equation}
g(\cdot ; \pi) = \pi_0\delta_0(\cdot) + \sum\limits_{k = 1}^K \pi_k N(\cdot ; 0, \sigma_k^2) \ ,
\label{eq:ash_prior}
\end{equation}
where $\pi_0$ is the null proportion.
Here the mixture proportions $\pi \defeq \{\pi_0, \pi_1, \ldots, \pi_K\}$ are non-negative and sum to 1, and are to be estimated,
whereas the component standard deviations
$\sigma_1 < \sigma_2 < \cdots < \sigma_K$ are a fixed pre-specified grid of values. By using a sufficiently wide and dense grid of standard deviations this finite mixture can approximate, to arbitrary accuracy, any scale mixture of zero-mean Gaussians.

The marginal log-likelihood for $\pi,\omega$,
integrating out $\theta$, $Z$, is given by the following Theorem.
\begin{theorem} \label{theorem:marginal}
Combining \eqref{eq:xcond}-\eqref{eq:ash_prior}, the marginal log-likelihood of $\pi,\omega$ is
\begin{equation}
L(\pi, \omega) \defeq
\log\left(\prod\limits_{j = 1}^n p(X_j | \pi, \omega)\right)
=\sum\limits_{j = 1}^n
\log\left(
\sum\limits_{k = 0}^K \pi_k
\left(
p_{jk0} +
\sum\limits_{l = 1}^L \omega_l p_{jkl}
\right)
\right), 
\label{eq:cash_loglik}
\end{equation}
where
\begin{equation}
p_{jkl} = \frac{s_j^l}{\sqrt{\sigma_k^2 + s_j^2}^{l+1}}
\frac{1}{\sqrt{l!}}
\varphi^{(l)}\left(\frac{
X_j
}{
\sqrt{\sigma_k^2 + s_j^2}
}\right).
\end{equation}
\end{theorem}
\begin{proof}
See Appendix \ref{sec:app}.
\end{proof}

\subsection{Fitting the model} \label{sec:cash_optim}

Following the usual EB approach, we fit the model \eqref{eq:xcond}-\eqref{eq:ash_prior} 
in two steps, first estimating $g,f$ by estimating $\pi,\omega$ and then basing inference for $\theta$ on the (estimated) posterior distribution $p(\theta | X , s, \hat\pi,\hat\omega)$. Note that under the model \eqref{eq:xcond}-\eqref{eq:ash_prior} $\theta_1,\ldots,\theta_p$  are conditionally independent given $f,g,X,s$, so this posterior distribution $p(\theta | X , s, \hat\pi,\hat\omega)$ factorizes, and is determined by its marginal distributions $p(\theta_j | X_j, s_j, \hat\pi,\hat\omega)$. The intuition here is that, under the exchangeability assumption, the effects of correlation are captured entirely by the (realized) correlated noise distribution $f$. Once this distribution is estimated, the inferences for each $\theta_j$ become independent, just as in the standard EBNM problem.

The usual EBNM approach to estimating $\pi,\omega$ would be to maximize the likelihood $L(\pi,\omega)$. Here we modify this approach using maximum penalized likelihood. Specifically we 
use the penalty on $\omega$ as in \eqref{eq:penalty_h}, and the penalty on $\pi$ 
used by \cite{stephens2017} to encourage conservative (over-)estimation of
the null proportion $\pi_0$ (to induce conservative estimation of false discovery rates). 
Thus, we solve
\begin{equation} \label{eq:cash_optim}
\hat{\pi},\hat{\omega}  = \arg \max\limits_{\pi, \omega}  \sum\limits_{j = 1}^n
\log\left(
\sum\limits_{k = 0}^K \pi_k 
\left(
p_{jk0}  +
\sum\limits_{l = 1}^L \omega_l p_{jkl}
\right)
\right)
+ \sum\limits_{k = 0}^K\lambda_k\log(\pi_k)
- \sum\limits_{l = 1}^L\gamma_l|\omega_l|
 \end{equation}
subject to the constraints
 \begin{align}
\sum\limits_{k = 0}^K\pi_k & = 1\\
\pi_k  & \ge 0, \ \ k = 0, 1, \ldots, K\\
\varphi\left(\mathfrak z_i\right) +
\sum\limits_{l = 1}^L \omega_l \frac{1}{\sqrt{l!}}\varphi^{(l)}(\mathfrak z_i) & \geq 0, \ \ i = 1, \ldots, m \ .
\label{eq:cash_nonneg}
\end{align}
In \eqref{eq:cash_nonneg} we used the same device as in
\eqref{prob:z_optim_simple} to capture non-negativity of $f$. 
We set $\gamma_l$ as in \eqref{eq:gamma}, use only the first $L=10$ Gaussian derivatives, and set $\lambda_0 = 10$, $\lambda_1 = \cdots = \lambda_K = 0$ as in \cite{stephens2017}.

Problem \eqref{eq:cash_optim} is biconvex. That is, given a feasible $\hat\pi$, the optimization over $\omega$ is convex; and given a feasible $\hat\omega$, the optimization over $\pi$ is convex. The optimization over $\pi$ 
can be solved using the EM algorithm, or more
efficiently using convex optimization methods \citep{koenker2014,koenker2017,mixsqp}. To optimize over $\omega$ we use the same approach as in solving \eqref{prob:z_optim_simple}. To solve \eqref{eq:cash_optim} we simply iterate between
these two steps until convergence.

\subsection{Posterior calculations}

For each $j$, the posterior distribution
$p(\theta_j \mid X_j , \hat\pi,\hat\omega)$ is, by Bayes Theorem, given by
\begin{equation}
p(\theta_j \mid X_j , \hat\pi,\hat\omega)
=
\frac{
\left[\hat\pi_0\delta_0 + \sum\limits_{k = 1}^K \hat\pi_k N(\theta_j ; 0, \sigma_k^2)\right]
\left[\frac{1}{s_j}
\varphi\left(\frac{X_j - \theta_j}{s_j}\right)
+
\sum\limits_{l = 1}^L \hat\omega_l
\frac{1}{s_j}
\frac{1}{\sqrt{l!}}
\varphi^{(l)}\left(
\frac{X_j - \theta_j}{s_j}
\right)\right]
}{
\sum\limits_{k = 0}^K \hat\pi_k
\left(
p_{jk0} +
\sum\limits_{l = 1}^L \hat\omega_l p_{jkl}
\right)
} .
\end{equation}
Despite the somewhat complex form, some important functionals of this posterior distribution are analytically
available.
\begin{enumerate}
    \item The posterior mean for $\theta_j$
\begin{equation}
E[\theta_j \mid X_j , \hat\pi,\hat\omega]
=
\displaystyle
\frac{
\sum\limits_{k = 0}^K \hat\pi_k
\left(
m_{jk0} +
\sum\limits_{l = 1}^L \hat\omega_l m_{jkl}
\right)
}{
\sum\limits_{k = 0}^K \hat\pi_k
\left(
p_{jk0} +
\sum\limits_{l = 1}^L \hat\omega_l p_{jkl}
\right)
} \ ,
\end{equation}
where $m_{jkl} =
-
\frac{s_j^l \sigma_k^2}{\sqrt{\sigma_k^2 + s_j^2}^{l+2}}
\frac{1}{\sqrt{l!}}
\varphi^{(l+1)}\left(\frac{X_j}{\sqrt{\sigma_k^2 + s_j^2}}\right)
$.

\item The local FDR \citep[lfdr;][]{efron2008}  is
\begin{equation}
\text{lfdr}_j \defeq {\text{Pr}}(\theta_j = 0 \mid X_j , \hat\pi,\hat\omega)
=
\frac{
\hat\pi_0
\frac{1}{s_j}
\varphi\left(\frac{X_j}{s_j}\right)
+
\sum\limits_{l = 1}^L \hat\omega_l
\frac{1}{s_j}
\frac{1}{\sqrt{l!}}
\varphi^{(l)}\left(\frac{X_j}{s_j}\right)
}{
\sum\limits_{k = 0}^K \hat\pi_k
\left(
p_{jk0} +
\sum\limits_{l = 1}^L \hat\omega_l p_{jkl}
\right)
}.
\end{equation}
From this, the FDR of any discovery set $\Gamma\subseteq\{1,\ldots,n\}$ can be estimated as
\begin{equation}
\widehat{\text{FDR}}(\Gamma)
= 
\frac{1}{|\Gamma|}
\sum\limits_{j\in\Gamma}
\text{lfdr}_j \ ,
\end{equation}
where $|\Gamma|$ denotes the number of elements in $\Gamma$. Storey's $q$-value \citep{storey2003} for each $j$ is defined as
\begin{equation}
q_j \defeq \widehat{\text{FDR}}(\{k: \text{lfdr}_k \le \text{lfdr}_j\}) \ .
\end{equation}

\item \cite{stephens2017} introduced the term ``local false sign rate (lfsr)" to refer to the probability of getting the sign of an effect wrong, as well as the false sign rate (FSR) and the $s$-value, analogous to the FDR and the $q$-value, respectively. 
Making statistical inference about the sign of a parameter, rather than solely focusing on whether the parameter being zero or not, was also discussed in \cite{tukey1991,gelman2012}.
The value of $\text{lfsr}_j$ is defined as
\begin{equation}
\text{lfsr}_j \defeq \min\{
\text{Pr}(\theta_j \geq 0 \mid X_j, \hat \pi, \hat \omega), \ 
\text{Pr}(\theta_j \leq 0 \mid X_j, \hat \pi, \hat \omega)
\} \ ,
\end{equation}
which is easily calculated from $\text{lfdr}_j$ and
\begin{equation}
\text{Pr}(\theta_j > 0 \mid X_j, \hat \pi, \hat \omega)
=
\frac{
\sum\limits_{k = 1}^K \hat\pi_k
\left(
\hat\tau_{jk0} +
\sum\limits_{l = 1}^L \hat\omega_l \tau_{jkl}
\right)
}
{
\sum\limits_{k = 0}^K \hat\pi_k
\left(
p_{jk0} +
\sum\limits_{l = 1}^L \hat\omega_l p_{jkl}
\right)} \ ,
\end{equation}
where $\tau_{jkl} = \frac{s_j^l
}
{\sqrt{l!}\sqrt{s_j^2 + \sigma_k^2}^{l + 1}}
\left(
\sum\limits_{m = 0}^{l}
\binom{l}{m}
\left(\frac{\sigma_k}{s_j}\right)^m
\varphi^{(m - 1)}
\left(\frac{X_j}{\sqrt{s_j^2 + \sigma_k^2}}\frac{\sigma_k}{s_j}\right)
\varphi^{(l - m)}
\left(\frac{X_j}{\sqrt{s_j^2 + \sigma_k^2}}\right)
\right)$. The FSR and $s$-value are estimated and defined similarly to the FDR and $q$-value as
\begin{equation}
\widehat{\text{FSR}}(\Gamma)
= 
\frac{1}{|\Gamma|}
\sum\limits_{j\in\Gamma}
\text{lfsr}_j
\ , \qquad
s_j \defeq \widehat{\text{FSR}}(\{k: \text{lfsr}_k \le \text{lfsr}_j\}) \ .
\end{equation}

\end{enumerate}




\subsection{Software}

We implemented both the fitting procedure and posterior calculations in an \R{} package \cash{} which is available at \url{https://github.com/LSun/cashr}. For $p = 10^4$, it takes on average about 6 seconds for model fitting and posterior calculations on a personal computer (Apple iMac, 3.2 GHz, Intel Core i5).

\section{Numerical Results} \label{sec:examples}

We now empirically assess the performance of \cash{} on both simulated and real data. We focus our assessments on the ``multiple testing" setting where $\theta$ is sparse and
the main goal is to identify ``significant" non-zero elements $\theta_j$. This problem
can be tackled using EB methods \citep{thomas1985,greenland1991} and here
we compare \cash{} with both \locfdr{} \citep{locfdr}, which 
attempts to capture effects of correlation through an empirical null strategy discussed in Section \ref{sec:nongaussian}, and \ash{} \citep{stephens2017}, which fits the same EBNM model as \cash{} but without allowing for correlation -- i.e. \ash{} is equivalent to setting $f=\varphi$ in \eqref{eq:ecsn_z}. Multiple testing can also be
tackled by attempting to control the FDR in the frequentist sense, and so we also compare with the Benjamini-Hochberg procedure \citep[BH;][]{benjamini1995} and \qvalue{} \citep{storey2002,storey2003}.
One advantage of the EBNM approach to multiple testing is that it can provide not only FDR assessments, but also point estimates and interval estimates for the effects $\theta_j$ \citep{stephens2017}. However, to keep our comparisons simple we focus here only on FDR assessments.


\subsection{Realistic simulation with gene expression data} \label{sec:sim}

We constructed synthetic data with realistic correlation structure using the simulation framework in Section \ref{sec:distortion}. The data are simulated according to the EBNM with correlated noise model \eqref{eq:g}-\eqref{eq:ebnm} as follows.
\begin{itemize}
\item The $p = 10^4$ normal means $\theta_1,\dots,\theta_p$ are iid samples from
    \begin{equation}
   g(\cdot) =\pi_0\delta_0(\cdot) + (1 - \pi_0)g_1(\cdot) \ ,
    \end{equation}
for six choices of $g_1$ and three choices of $\pi_0 \in \{0.5, 0.9, 0.99\}$ (Figure \ref{fig:fdr_g}). The density functions of these six choices of $g_1$ and other simulation details are in Appendix \ref{sec:sim_detail}.
\item To make the correlation structure among noise realistic, in each simulation $Z$ are simulated from real gene expression data as in Section \ref{sec:distortion}.
\item The standard deviations $s$ are also simulated from real gene expression data using the same pipeline, and are scaled to have $\frac1p\sum s_j^2 = 1$.
\item The observations are constructed as $X_j = \theta_j + s_j Z_j$, $j = 1, \ldots, p$.
\end{itemize}
In each simulated data set, this framework generates $p$ correlated observations $X_j$ of respective normal means $\theta_j$ with corresponding standard deviations $s_j$. The data $\{(X_1, s_1), \ldots, (X_p, s_p)\}$ are made available to each method, while the effects $\theta_j$ are withheld. The analysis goal is to identify which $\theta_j$ are significantly different from 0.  We applied each method to formulate a discovery set at nominal FDR $=0.1$, and calculated the empirical false discovery proportion (FDP) for each discovery set. We ran 1000 simulations for each $g_1$, divided evenly among the three choices of $\pi_0$. 


\begin{figure*}[!htb]
\begin{center}
\includegraphics[width = \linewidth]{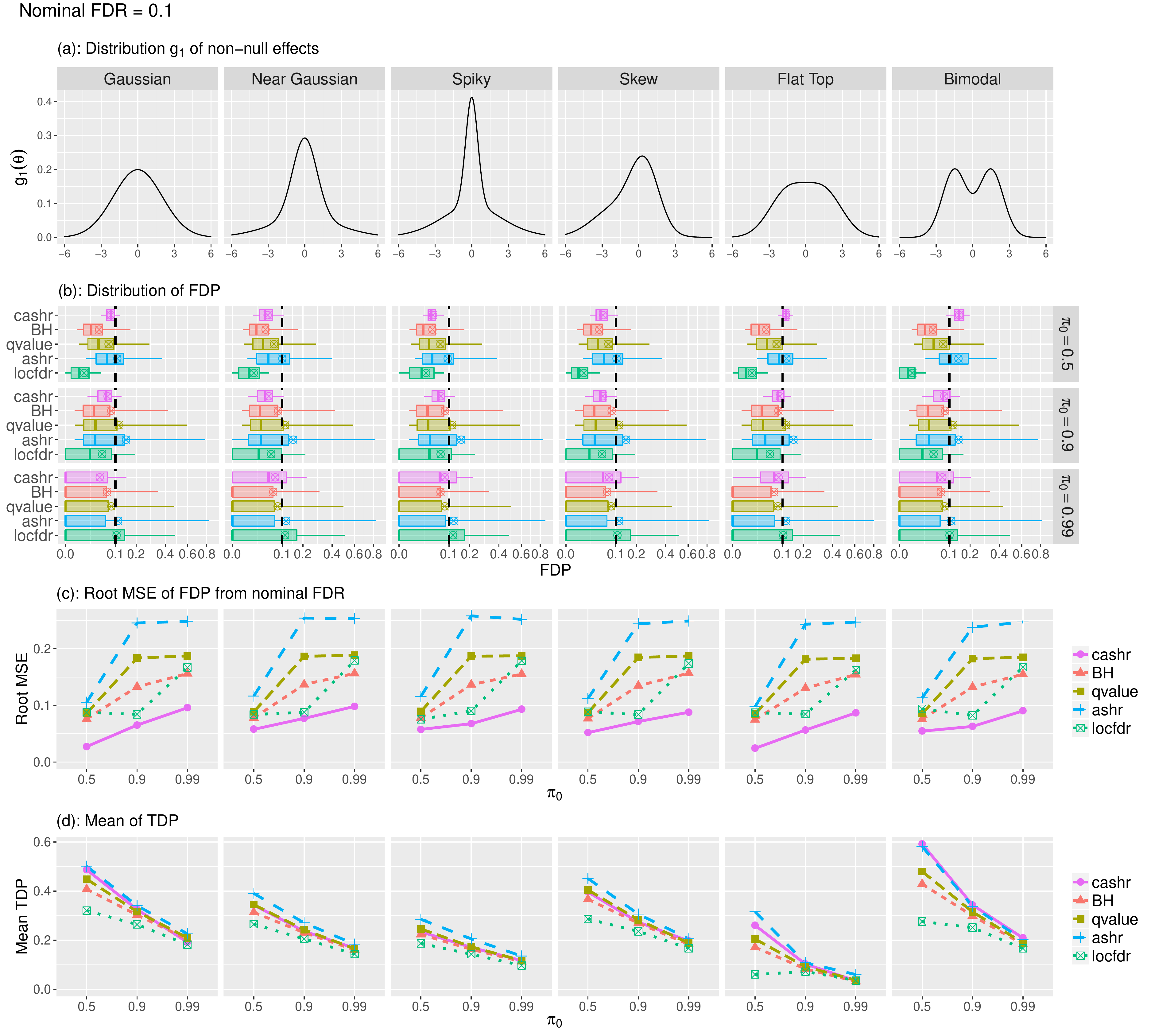}
\caption{Illustration that \cash{} outperforms other methods in producing discovery sets
whose FDP are consistently close to the nominal FDR, while maintaining good statistical power.
Simulation results are shown for six different distributions for the non-null effect ($g_1$; panel (a)) and three different values of the null proportion ($\pi_0 \in \{0.5,0.9,0.99\}$), stratified by methods. Panel (b): Comparison of the distribution of FDP, summarized as boxplots on square-root scale. The boxplots show the mean (cross), median (line), inter-quartile ranges (box), and 5th and 95th percentiles (whiskers). Panel (c): Comparison of the root MSE of FDP from the nominal FDR of 0.1, defined as $\sqrt{\text{mean}[(\text{FDP}-0.1)^2]}$. In all scenarios the distribution of FDP for \cash{} is more concentrated near the nominal $0.1$ level than other methods. Especially, the root MSE of FDP for \cash{} is uniformly lower than other methods. Panel (d): Comparison of the mean of TDP, as an indication of statistical power. On average, \cash{} maintains good power, only worse than \ash{} in some scenarios, which sometimes finds more true signals at the cost of severely losing control of FDP.
}
\label{fig:fdr_g}
\end{center}
\end{figure*}

Figure \ref{fig:fdr_g} compares the performance of each method in these simulations. Our first result is that, despite the presence of correlation, most of the methods control FDR in the usual frequentist sense under most scenarios: that is, the mean FDP is usually below the nominal level of 0.1. Indeed, BH is notable in never showing a mean FDP exceeding the nominal level, even though, as far as we are aware, no known theory guarantees this under the realistic patterns of
correlation used here (\cite{benjamini2001} gives relevant theoretical results under more restrictive assumptions on the correlation). The method most prone to lose control is \ash{}, but even its mean FDP is never above 0.2. 

However, despite this frequentist control of FDR, for most methods the FDP for individual data sets can often lie far from the nominal level \citep[see also][for example]{owen2005,qiu2005,blanchard2009,friguet2009}. Arguably, then, frequentist control of FDR is insufficient in practice, since we desire -- as far as is possible -- to make sound statistical inference for each data set.
That is, we might consider
a method to perform well if its FDP is consistently close to the nominal level, rather than close on average. By this criterion, \cash{} consistently outperforms other methods (Figure \ref{fig:fdr_g}): it provides uniformly lower root MSE of FDP from the nominal FDR, $0.1$, and the whiskers in the boxplots (indicating 5th and 95th percentiles) are narrower. Along with FDP, Figure \ref{fig:fdr_g} also shows the empirical true discovery proportion (TDP), defined as the proportion of true discoveries out of the number of all non-zero $\theta_j$, as an indication of statistical power. \cash{} maintains good power in that it produces higher TDP than most methods in most scenarios.
In some scenarios, \ash{} sometimes finds more true signals than \cash{}, but at the cost of severely losing control of FDP.

We note that \cash{} performs well even in settings that do not fully satisfy its underlying assumptions (e.g. where $g_1$ is asymmetric or multimodal). Note also that for our choices of $g_1$, $\pi_0 = 0.99$ is a highly sparse setting, as a large portion of the non-zero $\theta_j$ are close to zero. For example, when $g_1$ is Gaussian, only about $3$ out of $10^4$ $|\theta_j|$ are expected to be larger than $\sqrt{2\log p} \approx 4.3$. Therefore, it is understandable that no methods perform particularly well in this difficult setting. But even for this $\pi_0 = 0.99$ setting, although first impressions from the plot may be that \cash{} and BH perform similarly, closer visual inspection
shows $\cash{}$ to be better, in that its median FDP tends to be closer to $0.1$.

\begin{figure*}[!htb]
\begin{center}
\includegraphics[width = 0.90\linewidth]{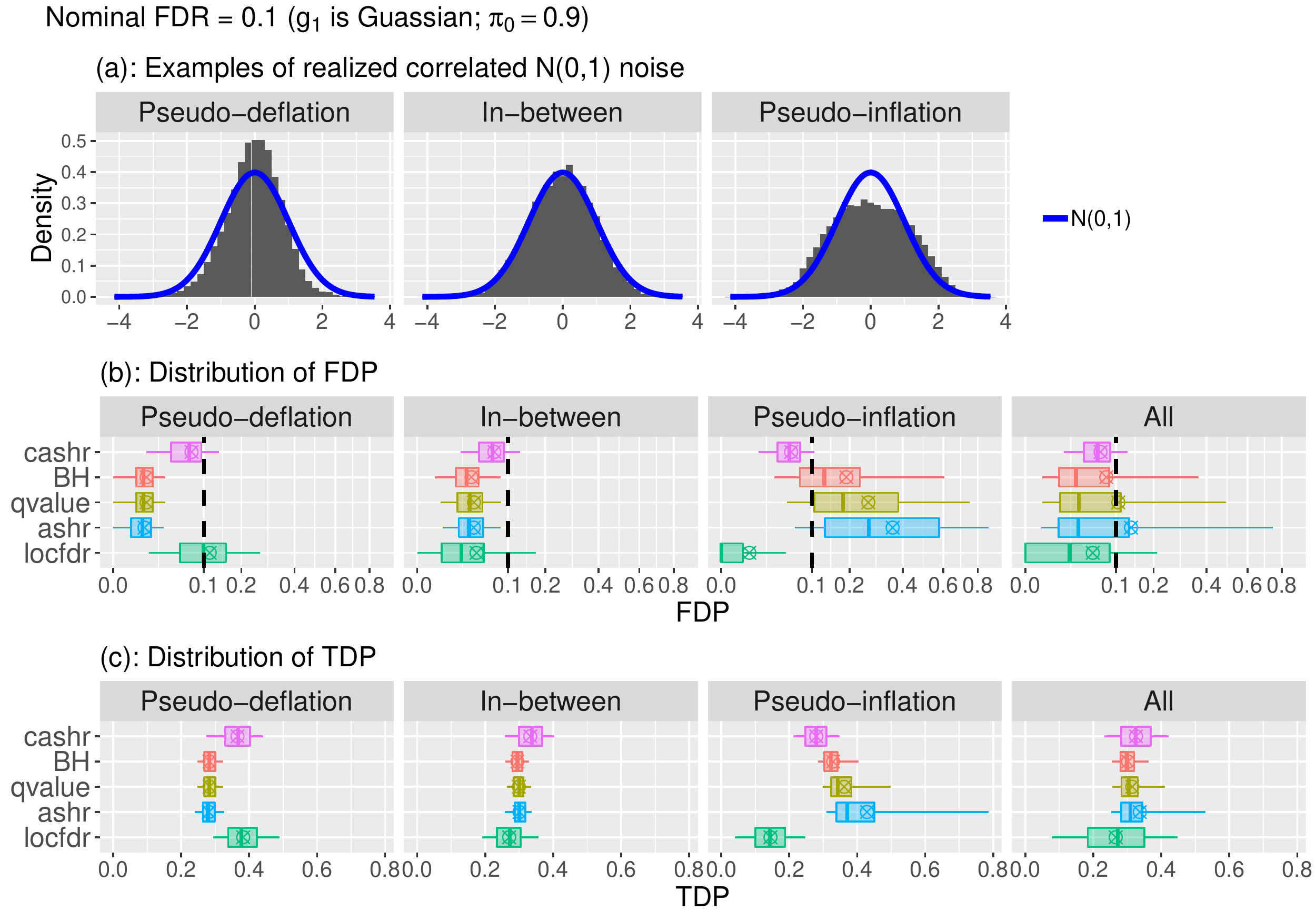}
\caption{Illustration that 
\cash{} consistently produces reliable FDP under different types of correlation-induced distortion. Here we take the results from a single simulation scenario ($g_1$ is Gaussian, $\pi_0=0.9$)
and stratify them into three groups of equal size according to the sample standard deviations of the realized correlated $N(0, 1)$ noise. Methods that ignore correlations among observations (BH, \qvalue{}, \ash{}) are generally too conservative under pseudo-deflation and too anti-conservative under pseudo-inflation; \locfdr{} tends to be too conservative under pseudo-inflation and consequently lose power; \cash{} maintains good FDR control in all settings. The boxplots show the mean (cross), median (line), inter-quartile ranges (box), and 5th and 95th percentiles (whiskers). FDP are plotted on square-root scale. Other choices of $g_1$ and $\pi_0$ give qualitatively similar results (not shown here).}
\label{fig:fdp_noise_sep}
\end{center}
\end{figure*}

The reason that \cash{} produces more consistently reliable  FDP is that, by design, it adapts itself to the particular correlation-induced distortion present in each data set. As illustrated in Figure \ref{fig:cor_z}, correlation can lead to pseudo-inflation in some data sets and pseudo-deflation in others. \cash{} is able to recognize which pattern is present, and correspondingly modify its behavior -- becoming more conservative in the former case and less conservative in the latter. This is illustrated in Figure \ref{fig:fdp_noise_sep}, which stratifies the realized data sets according to sample standard deviation of the realized correlated $N(0,1)$ noise $Z$ in each data set (for the setting where $g_1$ is Gaussian, $\pi_0=0.9$). The bottom $1/3$ are categorized as pseudo-deflation, top $1/3$ pseudo-inflation, and the others ``in-between.''

For data sets where $Z$ show no strong distortion (``in-between'') all methods give similar and reasonable results, with \cash{} showing only a small improvement.
However, when $Z$ are pseudo-inflated, methods ignoring correlation, such as BH, \qvalue{}, \ash{}, tend to be anti-conservative; that is, they form discovery sets whose FDP are often much larger than the nominal FDR. In contrast, \cash{} produces conservative FDP near the nominal value; and \locfdr{} is too conservative, consequently losing substantial power (discussed further in Section \ref{sec:leuk_mouse}). Conversely, with pseudo-deflation, methods ignoring correlation are too conservative, 
 producing FDP much smaller than the nominal FDR, losing power compared with \cash{} and \locfdr{}.

\subsection{Real data illustrations} \label{sec:leuk_mouse}

We now use two real data examples to illustrate some of the features of \cash{} (and other methods)
that we observed in simulated data. The first example is a well-studied data set from a leukemia study \citep{golub1999}, comparing gene expression in $47$ acute myeloid leukemia vs $25$ acute lymphoblastic leukemia samples, which was discussed extensively in \cite{efron2010} as a prime example of how correlation can distort empirical distributions. The second example comes from a study on embryonic mouse hearts \citep{smemo_thesis2012}, comparing gene expression in 2 left ventricle samples vs 2 right ventricle samples. 
(The number of samples is small, but each sample is a pool of ventricles from 150 mice -- necessary to obtain sufficient tissue for the experiments to work well -- and so this experiment involved dissection of 300 mouse hearts.)

For each data set we let $\theta_j$ denote the true $\log_2$-fold change in gene expression between the two groups for gene $j$. We use a standard analysis protocol (based on \cite{smyth2004}; see Appendix \ref{sec:sim_detail} for details) to obtain an estimate $X_j$ for $\theta_j$, and a corresponding $p$-value $p_j$. As in Section \ref{sec:distortion}, we convert the $p$-value to the corresponding $z$-score $z_j$ and use this to compute an effective standard deviation $s_j$.

Figure \ref{fig:leuk_mouse_hist} shows the empirical distribution of the $z$-scores for each data set, together
with the fitted correlated noise distribution from \cash{} and the fitted empirical null from \locfdr{}.
In both cases the histogram is substantially more dispersed than $N(0,1)$. However the two data sets
have otherwise quite different patterns of inflation: the leukemia data show inflation in both
the shoulders and tails of the distribution, whereas the mouse data show inflation only in the shoulders. This indicates the presence of some strong signals in the leukemia data, whereas the inflation in the mouse data  may
be primarily pseudo-inflation caused by correlation. Consistent with this, both \locfdr{} and \cash{} identify hundreds of significant signals in the leukemia data (at nominal FDR $=0.1$), but no significant signals in the mouse data (Table \ref{table:leuk_mouse}).

\begin{figure*}[!htb]
\begin{center}
\includegraphics[width = 0.49\linewidth]{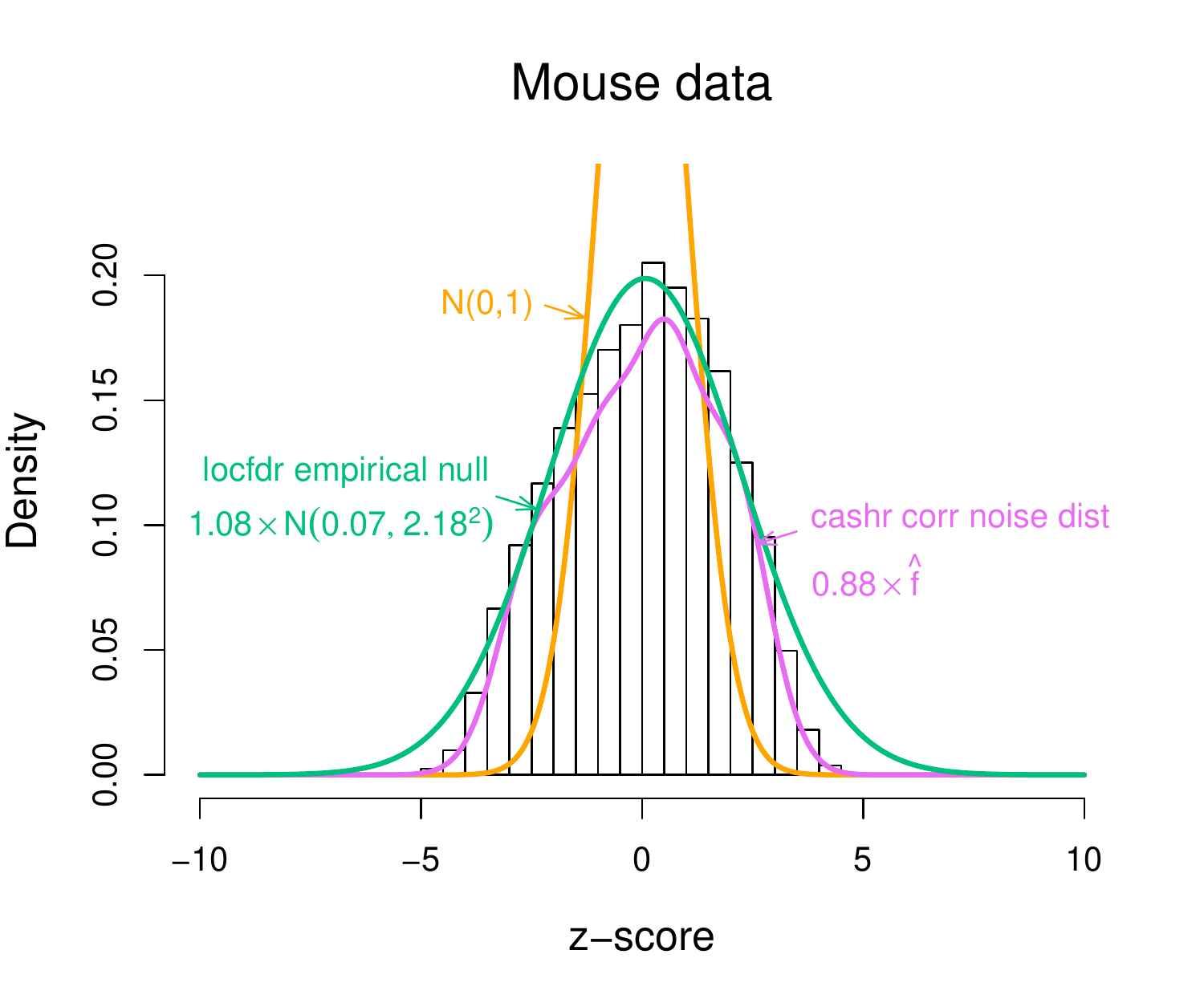}
\includegraphics[width = 0.49\linewidth]{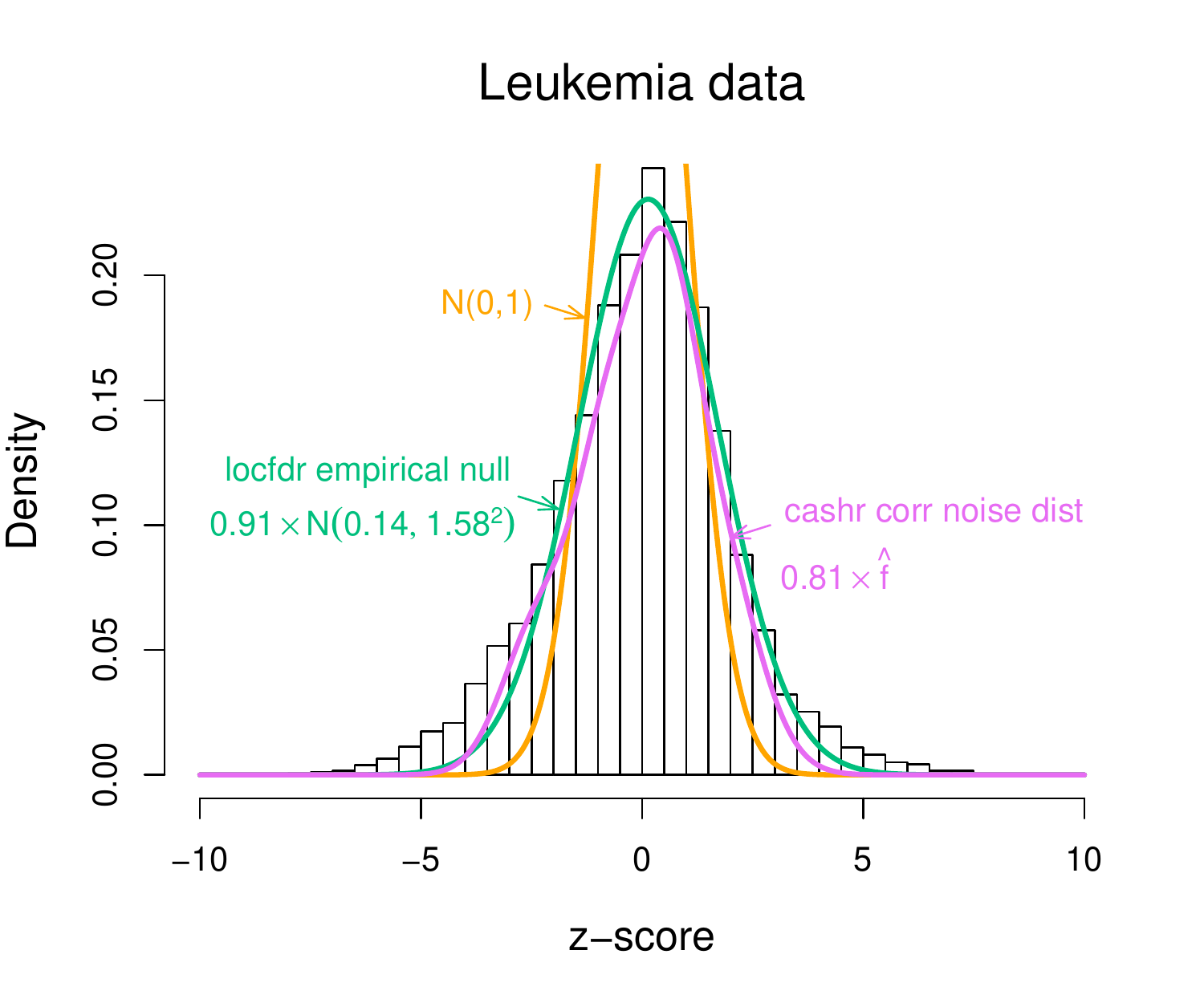}
\caption{Distribution of $z$-scores from analyzing gene differential expression in two real data sets. In both data sets, for each gene $j$, a $z$-score $z_j$ is computed, and $z_j \sim N(0, 1)$ under the null hypothesis of no differential expression. Then we compare the histogram of $z$-scores with $N(0,1)$, the fitted correlated noise distribution from \cash{}, and the fitted empirical null from \locfdr{}, scaled by respective estimated null proportions. Both histograms are substantially more dispersed than $N(0,1)$. The mouse data show inflation primarily in the shoulders of the distribution, and the fitted correlated noise distribution from \cash{} appears to be a much better fit than the fitted empirical null from \locfdr{}, particularly in the tails. The leukemia data show inflation in both shoulders and tails of the distribution, indicating the presence of some strong signals. Although otherwise similar, the fitted correlated noise distribution from \cash{} has a noticeably shorter right tail than the fitted empirical null from \locfdr{}, improving power.}
\label{fig:leuk_mouse_hist}
\end{center}
\end{figure*}

\begin{table}[!htb]
\centering
\begin{tabular}{| c | c | c |}
\hline
& \multicolumn{2}{c|}{Number of discoveries} \\
\hhline{|~|-|-|}
\hspace {0.0cm}Method\hspace {0.0cm} & \hspace {0.0cm} Leukemia data \hspace {0.0cm} & \hspace {0.25cm} Mouse data \hspace {0.25cm} \\
\hline
\cash{} & 385 & 0 \\
\locfdr{} & 282 & 0 \\
\hdashline
BH & 1579 & 4130\\
\qvalue{} & 1972 & 6502 \\
\ash{} & 3346 & 17191 \\
\hline
\end{tabular}
\caption{Numbers of discoveries from different methods at nominal FDR $= 0.1$. We analyzed $7128$ genes in the leukemia data and $17191$ genes in the mouse data. In both data sets, the $z$-score distributions appear to have correlation-induced inflation, and the numbers of significant discoveries declared by methods accounting for correlation (\cash{} and \locfdr{}) are much smaller than those ignoring correlation (BH, \qvalue{}, \ash{}). For the leukemia data, \cash{} finds $37\%$ more significant genes than \locfdr{}.}
\label{table:leuk_mouse}
\end{table}

Although the conclusions from \cash{} and \locfdr{} are, here, qualitatively similar, there are some
notable differences in their results. First, in the mouse data, the \cash{} correlated noise distribution
gives, visually, a much better fit than the \locfdr{} empirical null, particularly in the tails (Figure \ref{fig:leuk_mouse_hist}).
This is because the \cash{} correlated noise distribution is ideally suited to capture this
``shoulder-but-not-tail'' inflation pattern that is symptomatic of correlation-induced inflation.
The Gaussian empirical null distribution assumed by \locfdr{} is simply inconsistent with these data.
Indeed, this inconsistency is reflected in the null proportion estimated by \locfdr{} (1.08) which 
exceeds the theoretical upper bound
of 1.

Second, in the leukemia data, \cash{} identifies
37\% more significant results than \locfdr{} (385 vs 282). This is consistent with the greater power of \cash{} vs \locfdr{} in our simulations. 
One reason that \locfdr{} can lose power is that its Gaussian empirical null distribution tends to overestimate inflation in the tails when it tries to fit inflation in the shoulders. We see this feature in the mouse data, and although less obvious, this
appears to also be the case for the leukemia data: 
the estimated standard deviation of the empirical null is 1.58, which is almost certainly too large: a pseudo-inflated Gaussian correlated noise distribution is unlikely to have standard deviation exceeding $1.4$ (Appendix \ref{sec:gauss_sgd}).
In comparison the fitted correlated noise distribution from \cash{} has a noticeably shorter right tail (e.g.~$z \in [4,5]$) which leads it to categorize more $z$-scores in the right tail as significant (Figure \ref{fig:leuk_mouse_hist}).
On a side note, \cash{} also experiences the benefits of \ash{} highlighted in \cite{stephens2017},
which can also help increase power. 
For example, the unimodal assumption on the effects -- which allows that some of the $z$-scores around zero may correspond to true, albeit non-significant, signals -- can help improve estimates of $\pi_0$, and hence improve power.

Another feature of \cash{}, which distinguishes it from \locfdr{}, is that, by estimating $g$ while accounting for correlation-induced distortion, it can provide an estimate on the effect size distribution, $g_1$. For the mouse data, \cash{} estimates $\hat\pi_0 = 0.88$, or 12\% of genes may be differentially expressed to some extent, although it is not able to pin down any clear example of a differentially expressed gene: no gene has an estimated local FDR less than 0.80. One possible explanation for the lack of significant results in this case is lack of power. However, the estimated $g_1$ from \cash{} suggests that there may simply not exist any large effects to be discovered: $99\%$ of the probability mass of the estimated $g_1$ is on effect size $\leq 0.26$, or a mere $1.2$-fold change in gene expression. Thus the signals here, if any, are too weak to be discerned from noise and pseudo-inflation.

 We also applied the other methods -- BH, \qvalue{}, and \ash{} -- to both data sets. All three methods find very large numbers of significant results in both data sets (Table \ref{table:leuk_mouse}). Although we do not know the truth in these real data, there is a serious concern that many of these results could be false positives, since these methods
 are all prone to erroneously viewing pseudo-inflation as true signal (Figure \ref{fig:fdp_noise_sep}), and Figure \ref{fig:leuk_mouse_hist} suggests that pseudo-inflation may be present in both data sets.
 

\section{Discussion} \label{sec:disc}

We have presented a general approach to accounting for correlations among observations in the widely-used Empirical Bayes Normal Means model. Our strategy exploits theoretical
results from \cite{schwartzman2010} to model the impact
of correlation on the empirical distribution of correlated
$N(0,1)$ variables, and convex optimization techniques
to produce an efficient implementation.  We demonstrated through empirical examples that this strategy can both improve estimation
of the underlying distribution of true effects (Figure \ref{fig:deconv}) and -- in the multiple testing setting -- improve estimation of FDR compared
with EB methods that ignore correlation (Figures \ref{fig:fdr_g}, \ref{fig:fdp_noise_sep}). To the best of our knowledge, \cash{} is the first EBNM methodology to deal with correlated noise in this way.

Our empirical results demonstrate some benefits
of the EB approach to multiple testing compared with
traditional methods. In particular, \cash{} provides, on average, more accurate estimates of the
FDP than either BH or \qvalue{}. However, although we find these empirical results
encouraging, we do not have theoretical guarantees
of (frequentist) FDR control. That said, theoretical guarantees
of FDR control under arbitrary correlation structure
are lacking even for the widely-studied BH method.
BH has been shown to control FDR under certain correlation stuctures \citep[e.g. ``positive regression dependence on subsets'';][]{benjamini2001}. The Benjamini-Yekutieli procedure \citep{benjamini2001} is proved to control FDR under arbitrary dependence, but at the cost of being excessively conservative, and is consequently rarely used in practice.

A key feature of \cash{} is that it requires
no information about the actual correlations among observations. This has the important
advantage that it can be applied
wherever EBNM methods that ignore correlation can be applied. On the other hand, when additional
information on correlations is available it clearly may
be helpful to incorporate it into analyses. Within our approach such
information could be used to estimate the moments of
the pairwise correlations, and thus inform estimates of $\omega$ in the correlated noise distribution $f(\cdot; \omega)$. 
Alternatively, one could take a more ambitious approach: explicitly model the whole $p \times p$ correlation matrix, and
use this to help inform inference  \citep[e.g.][]{benjamini2007,wu2008,sun2009,friguet2009,fan2012}. Modeling correlation is likely to provide more efficient inferences when it can be accurately achieved \citep{hall2010}. However, in many situations -- particularly involving small sample sizes -- reliably modeling correlation may be impossible. 
Under what circumstances this more ambitious approach produces better inferences could be one area for future investigation.

The main assumptions underlying \cash{} are that the
correlated noise is marginally $N(0, 1)$, and that
the standard deviations are reliably computed. In the multiple testing setting this corresponds to assuming that the
test statistics are (marginally) well calibrated. If these
conditions do not hold -- for example, due to failure
of asymptotic theory underlying test statistic computations, or due to confounding factors (such as batch effects in gene expression studies), then \cash{} could give unreliable results. Of course \cash{} is not unique in this regard -- methods like BH and \qvalue{} similarly assume that test statistics are well calibrated. Dealing with confounders in gene expression studies is an 
active area of research, and several approaches exist,
many of them based on factor analysis \citep[e.g.][]{leek2007, sun2012,ruv2,wang2017,gerard.ruv,mouthwash}. 
Again, the possibility of combining these ideas
with our methods could be a future research direction.





\section*{Appendix}

\appendix

\counterwithin{table}{section}
\counterwithin{equation}{section}
\counterwithin{figure}{section}
\section{The marginal distribution of the simulated null $z$-scores} \label{sec:marginal_N01}

Figures \ref{fig:rand_gene} and \ref{fig:avg_cdf} offer support for the claim that the $z$-scores simulated in Section \ref{sec:distortion} are marginally $N(0,1)$-distributed.

Figure \ref{fig:rand_gene} compares $z$-scores simulated as in Section \ref{sec:distortion} with $z$-scores simulated under a modified framework that removes gene-gene correlations, and with iid $N(0,1)$ samples.
The modified framework uses exactly the same simulation and analysis pipeline as the original framework of Section \ref{sec:distortion}, with one important difference: in each simulation, {\it for each gene independently} we randomly selected two groups of five samples without replacement, hence removing gene-gene correlations. 

The empirical CDF of $10^4$ data sets simulated as in Section \ref{sec:distortion} show a huge amount of variability (panel (a)), presumably due to correlations among genes.  In the modified framework, correlation-induced distortion disappears: the empirical CDF of all $10^4$ data sets are almost exactly the same as $N(0,1)$ (panel (b)), just as with the iid $N(0,1)$ samples (panel (c)). This demonstrates that without gene-gene correlations, the analysis pipeline used here produces uncorrelated $N(0,1)$ $z$-scores.

\begin{figure*}[!htb]
\begin{center}
\includegraphics[width = 0.3\linewidth]{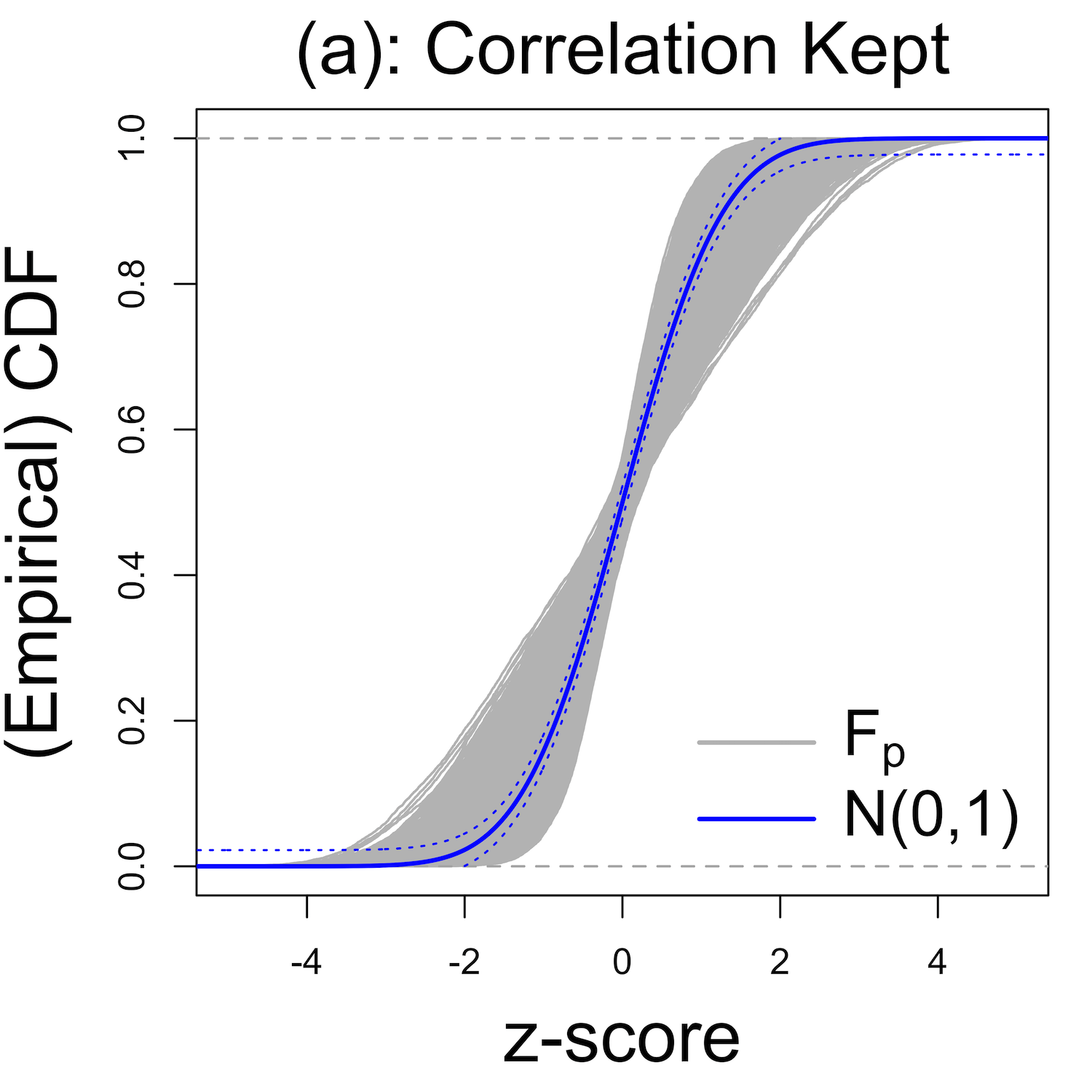}
\includegraphics[width = 0.3\linewidth]{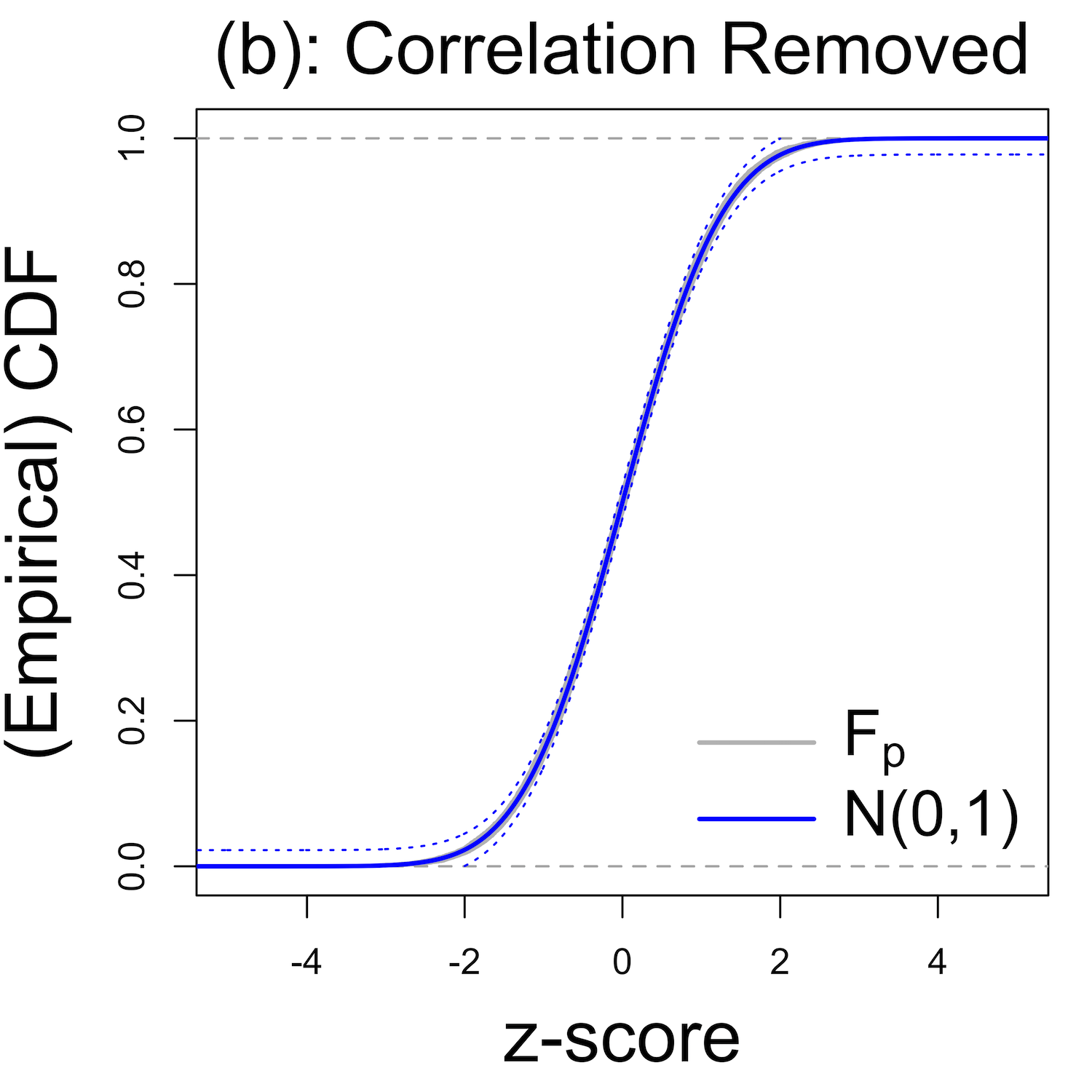}
\includegraphics[width = 0.3\linewidth]{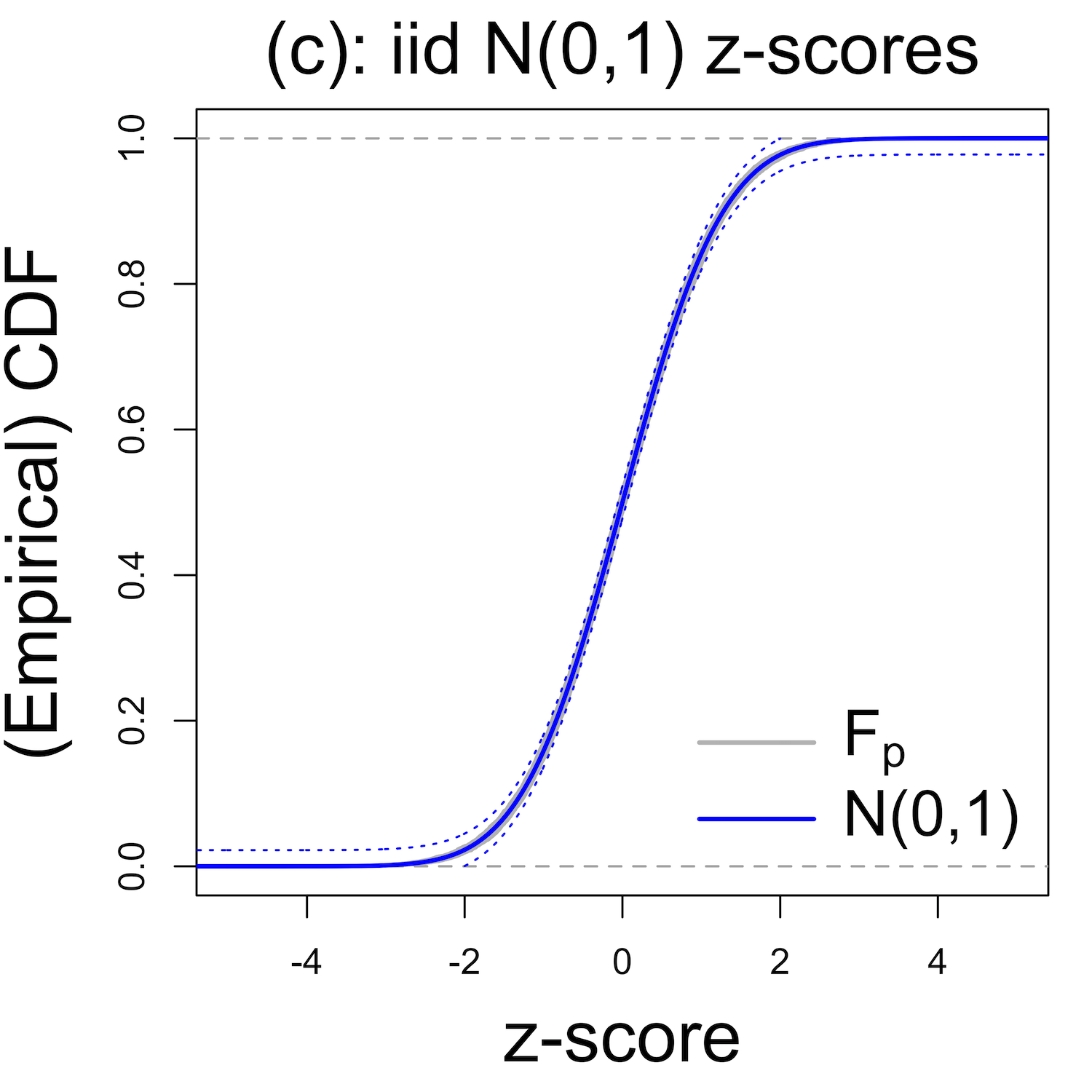}
\caption{Comparison of $10^4$ empirical CDF of $z$-scores ($F_p$) obtained by applying the same analysis pipeline to data simulated by two different frameworks: the original framework in Section \ref{sec:distortion} which keeps gene-gene correlations (panel (a)); and the modified framework to remove gene-gene correlations by randomizing samples for each gene (panel (b)). We also plot $10^4$ empirical CDF of iid $N(0,1)$ samples for comparison (panel (c)). The $z$-scores obtained under the original framework show clear correlation-induced distortion -- the variability of empirical CDF is huge. In contrast, when gene-gene correlations are removed under the modified framework, distortion disappears: empirical CDF are almost exactly the same as $N(0,1)$ and the variability is essentially invisible; indeed, they are indistinguishable from $10^4$ empirical CDF of iid $N(0,1)$ $z$-scores. It shows clear evidence that the analysis pipeline can produce well-calibrated null $z$-scores if no gene-gene correlations. Dotted lines are Dvoretzky-Kiefer-Wolfowitz bounds with $\alpha = 1 / 10^4$.}
\label{fig:rand_gene}
\end{center}
\end{figure*}

In addition, Figure \ref{fig:avg_cdf} shows that the mean empirical CDF of the $10^4$ data sets simulated from the original framework -- the average of empirical CDF of Figure \ref{fig:rand_gene}(a) -- is very close to $N(0,1)$. Possible deviation happens only in the far tails ($|\cdot| \in \{5, 6\}$). Compared with $N(0, 1.05^2)$ and $N(0, 1.1^2)$, the deviation is very small even on the logarithmic scale (panels (b-c)), probably caused by numerical constraints as one or two $z$-scores in this area in a few data sets can make a visible difference.

\begin{figure*}[!htb]
\begin{center}
\includegraphics[width = 0.3\linewidth]{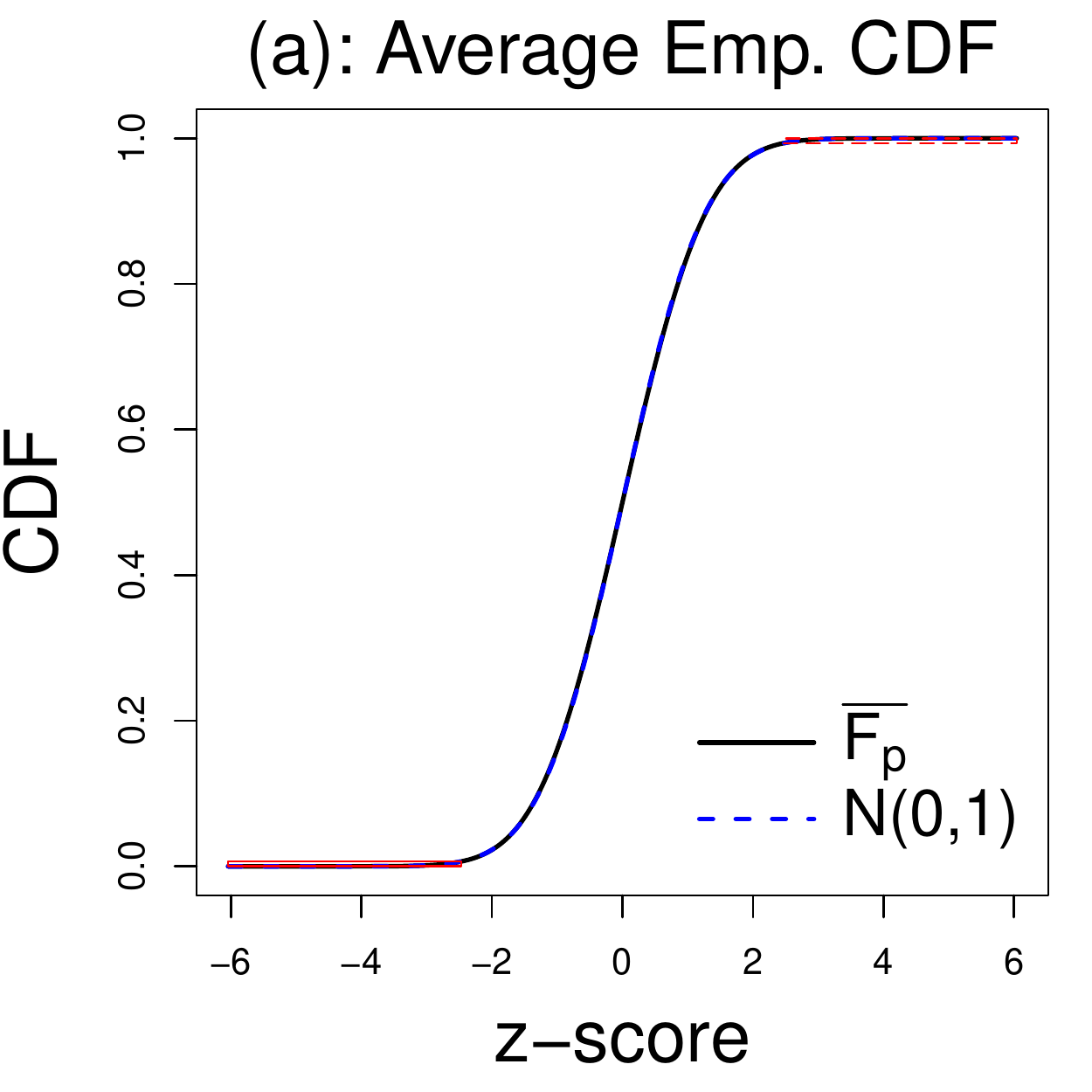}
\includegraphics[width = 0.3\linewidth]{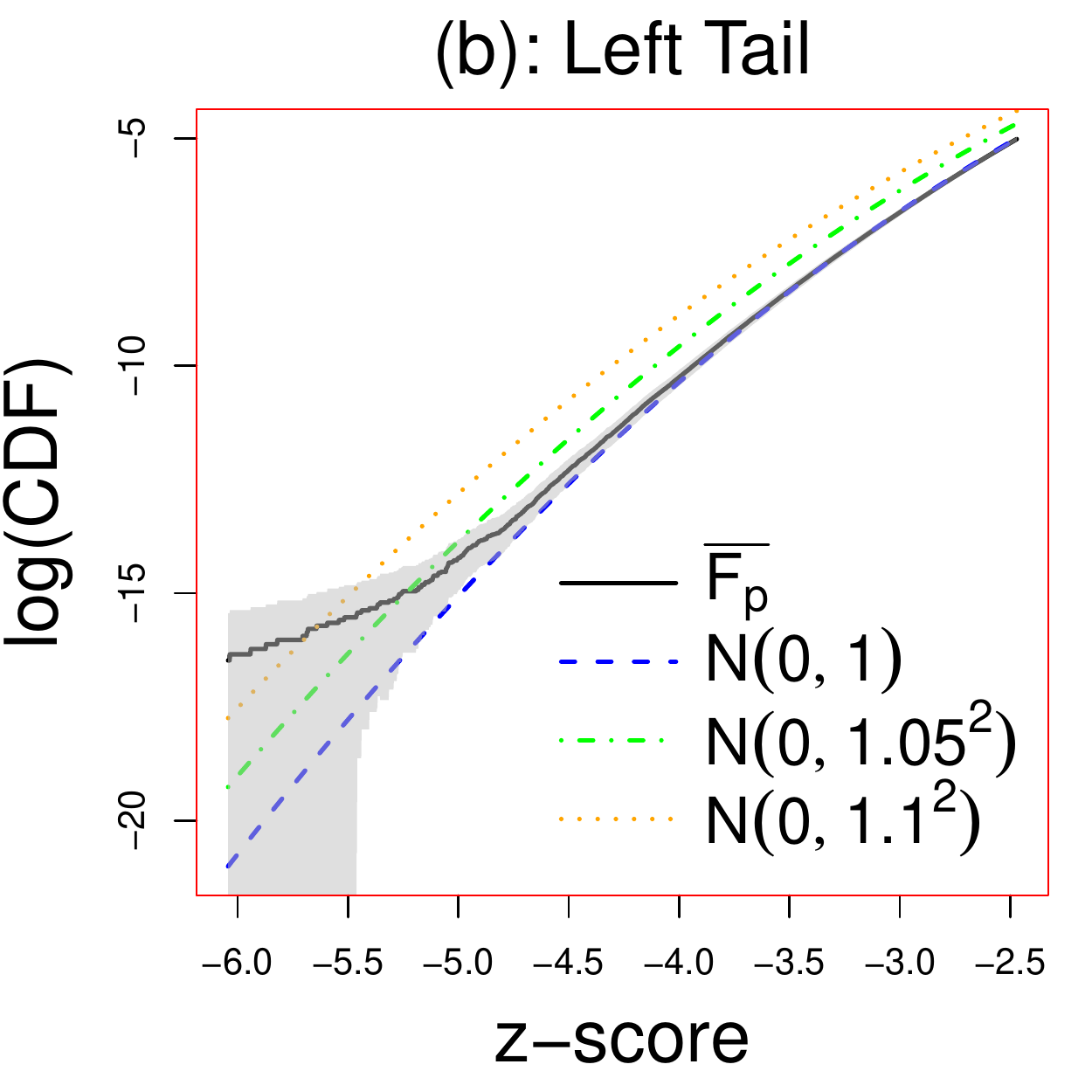}
\includegraphics[width = 0.3\linewidth]{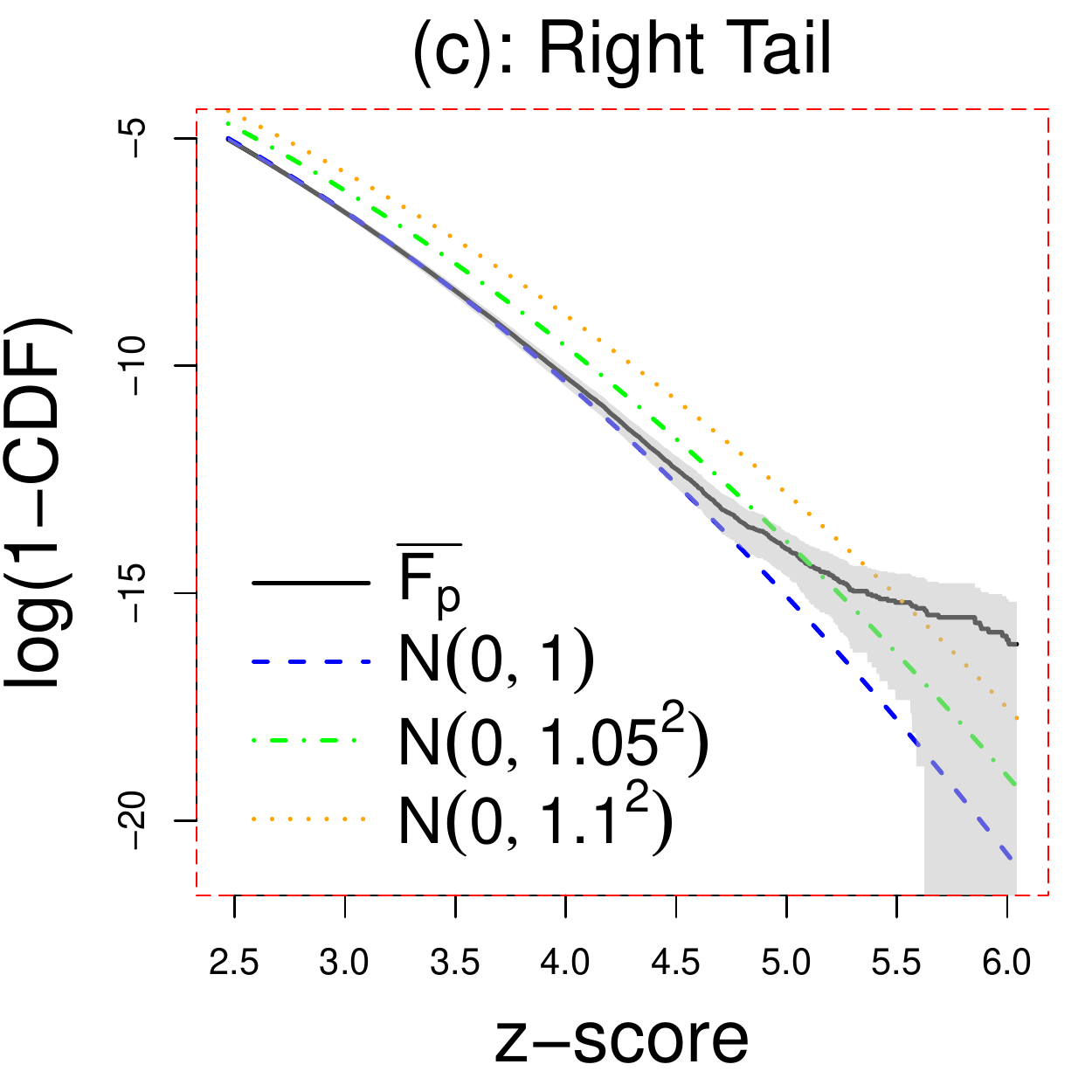}
\caption{Illustration that the average empirical CDF of $z$-scores ($\overline{F_p}$) simulated as in Section \ref{sec:distortion} closely matches $N(0,1)$, aggregated over $10^4$ data sets. Left: the average of all empirical CDF in Figure \ref{fig:rand_gene}(a). The average empirical CDF is extremely close to $N(0,1)$. Center and Right: the left and right tails of the average empirical CDF on logarithmic scale. Shaded areas are $99.9\%$ confidence bands. Compared with $N(0, 1.05^2)$ and $N(0, 1.1^2)$, possible deviation from $N(0,1)$ is light even in the far tails.}
\label{fig:avg_cdf}
\end{center}
\end{figure*}

\section{Decomposing Gaussian by standardized Gaussian derivatives} \label{sec:gauss_sgd}
\begin{proposition}
The PDF of $N(\mu, \sigma^2)$ can be decomposed by standard Gaussian and its derivatives in the form of \eqref{eq:gd_z_1} if and only if $\sigma^2 \leq 2$.
\end{proposition}
\begin{proof}
Let $h_l(\cdot)$ denote the $l^\text{th}$ probabilists' Hermite polynomial. The orthogonality and completeness of Hermite polynomials in $L^2(\mathbb{R}, \mathrm{d}\Phi)$ \citep[e.g.][]{szego1975} leads to the following fact
\begin{equation}
\int_\mathbb{R}
\frac{1}{\sqrt{m!}}h_m(x)\frac{1}{\sqrt{n!}}\varphi^{(n)}(x)\mathrm{d}x = (-1)^n\delta_{mn}\ , \qquad \forall m, n = 0, 1, 2, \ldots,
\end{equation}
where $\delta_{mn} = \begin{cases}
1 & m = n \\
0 & \text{otherwise}
\end{cases}$.
Therefore, if any PDF $f$ can be decomposed in the form of \eqref{eq:gd_z_1}, the coefficient of the $l^\text{th}$-order standardized Gaussian derivative has to be
\begin{equation} \label{eq:w_l}
w_l = (-1)^l\int_\mathbb{R}
\frac{1}{\sqrt{l!}}h_l(x)f(x)
\mathrm{d}x = \frac{(-1)^l}{\sqrt{l!}}E_{f}[h_l] \ ,
\end{equation}
where $E_{f}[h_l]$, sometimes called ``Hermite moment,'' is the expected value of $h_l(\cdot)$ when the PDF of the random variable is $f$.  If $f$ is $N(\mu,\sigma^2)$, we can obtain analytic expressions of these Hermite moments
\begin{equation} \label{eq:hermite_moments}
E_{f}[h_l] = 
\mu^l
+
\sum_{k = 1}^{\lfloor l / 2 \rfloor}
\binom{l}{2k}\mu^{l - 2k}(\sigma^2 - 1)^k(2k - 1)!!
\defeq
M_l(\mu,\sigma^2-1) \ ,
\end{equation}
where $n!!$ denotes the double factorial of $n$, and $M_l(x,y)$ denotes the function of $l^\text{th}$-order moment of a Gaussian with mean $x$ and variance $y$. Putting \eqref{eq:w_l}-\eqref{eq:hermite_moments} together, the coefficients in \eqref{eq:gd_z_1} become
\begin{equation}
w_l = \frac{(-1)^l}{\sqrt{l!}}M_l(\mu,\sigma^2-1) \ .
\end{equation}
Note that $w_l$ is not exploding if and only if $|\sigma^2 - 1|\leq1$ or equivalently, $\sigma^2\leq2$.
\end{proof}

This result suggests that a pseudo-inflated Gaussian correlated noise distribution is not likely to have standard deviation greater than $\sqrt{2} \approx 1.4$.

In the special case when $\rho_{ij} = 1$, $f$ becomes $\delta_z$, a point mass on $Z \equiv z$, with $z$ randomly sampled from $N(0, 1)$. It is interesting to note that $\delta_z$ can be decomposed in the form of \eqref{eq:gd_z_1} as
$$ 
\delta_z(\cdot) = \varphi(\cdot) + \sum\limits_{l = 1}^\infty
\left[\frac{(-1)^l}{\sqrt{l!}}h_l(z) \right]
\left[\frac{1}{\sqrt{l!}}\varphi^{(l)}(\cdot)\right] \ ,
\quad
\forall z \in \mathbb{R} \ .
$$

\section{Proof of Theorem \ref{theorem:marginal}} \label{sec:app}

\begin{proof}
The marginal distribution of $X_j$, denoted as $p(X_j)$, is obtained by integrating out $\theta_j$
\begin{align}
\displaystyle p(X_j) &
= \displaystyle \int_{\mathbb{R}} g(\theta_j) p(X_j|\theta_j, s_j) \mathrm{d}\theta_j 
= \displaystyle \int_{\mathbb{R}}
g(\theta_j)
\frac{1}{s_j}f\left(\frac{X_j - \theta_j}{s_j}\right)
\mathrm{d}\theta_j \nonumber
\\
&=\displaystyle\int_{\mathbb{R}}
\left[\pi_0\delta_0(\theta_j)
+
\sum\limits_{k = 1}^K\pi_k
\frac{1}{\sigma_k}\varphi\left(\frac{\theta_j}{\sigma_k}\right)\right]
\left[\frac{1}{s_j}
\varphi\left(\frac{X_j - \theta_j}{s_j}\right)
+\frac{1}{s_j}
\sum\limits_{l = 1}^L \omega_l
\frac{1}{\sqrt{l!}}
\varphi^{(l)}\left(
\frac{X_j - \theta_j}{s_j}
\right)\right] \mathrm{d}\theta_j \nonumber
\\
&=
\sum\limits_{k = 0}^K\pi_k
\left(
p_{jk0}
+
\sum\limits_{l = l}^L\omega_l p_{jkl}
\right) \ ,
\label{eq:cash_marginal}
\end{align}
where
$
p_{jkl} = \displaystyle \int_{\mathbb{R}}
\frac{1}{\sigma_k}\varphi\left(\frac{\theta_j}{\sigma_k}\right)
\frac{1}{s_j}
\frac{1}{\sqrt{l!}}\varphi^{(l)}\left(
\frac{X_j - \theta_j}{s_j}
\right) \mathrm{d}\theta_j
$ is essentially a convolution of $\varphi$ and $\varphi^{(l)}$ and has an analytic form
$$
\displaystyle p_{jkl} = 
\frac{s_j^l}{\sqrt{\sigma_k^2 + s_j^2}^{l+1}}
\frac{1}{\sqrt{l!}}
\varphi^{(l)}\left(\frac{
X_j
}{
\sqrt{\sigma_k^2 + s_j^2}
}\right) \ .
$$
This form is also valid for $k = 0$, $l = 0$.
Following \eqref{eq:cash_marginal}, the marginal log-likelihood of $\pi,\omega$ is given by
$$
\log\left(\prod\limits_{j = 1}^n p\left(X_j\right)\right)
=\sum\limits_{j = 1}^n\log\left(
\sum\limits_{k = 0}^K\pi_k
\left(
p_{jk0}
+
\sum\limits_{l = l}^L\omega_l p_{jkl}
\right)
\right) \ .
$$
\end{proof}

\section{Simulation details} \label{sec:sim_detail}

\subsection{Six choices of the non-null effect distribution}

Table \ref{table:g1} lists the details of the six choices of $g_1$, the non-null effects in Section \ref{sec:sim}. The table also shows the average signal strength, $E[|\cdot|^2]$, and the probability of large signal, $\text{Pr}(|\cdot|\ge\sqrt{2\log p})$, conditioned on $g_1$.

\begin{table}[!htb]
\centering
\caption{Details of $g_1$, the distribution of non-null effects in Section 4.1}\label{table:g1}
\begin{tabularx}{\textwidth}{c | X | c | c }
$g_1$ & PDF 
& $E[|\cdot|^2]$
&
$\text{Pr}(|\cdot|\ge\sqrt{2\log p}$) \\
\hline
Gaussian & $N(0, 2^2)$ & 4 & 0.032 \\
Near Gaussian & $0.6 N(0, 1) + 0.4 N(0, 3^2)$ & 4.2 & 0.061\\
Spiky & $0.4 N(0, 0.5^2) + 0.2 N(0, 2^2) + 0.4 N(0, 3^2)$ & 4.5 & 0.067\\
Skew & $0.25N(-2, 2^2) + 0.25N(-1, 2^2) + 0.25N(0, 1) + 0.25N(1, 1)$ & 4 & 0.045\\
Flat Top & $0.5N(-1.5, 1.5^2) + 0.5N(1.5, 1.5^2)$ & 4.5 & 0.031 \\
Bimodal & $0.5 N(-1.5, 1) + 0.5 N(1.5, 1)$ & 3.25 & 0.0026\\
\end{tabularx}
\end{table}


\subsection{Implementation of methods}

The existing methods we use for comparison in this paper mostly use the default settings in their respective \R{} packages. That include \REBayes{}, \deconvolveR{} for deconvolution (Section \ref{sec:distortion}), and \qvalue{}, \locfdr{} for multiple testing (Section \ref{sec:examples}). For \EbayesThresh{}, we set \texttt{a=NA} to allow the scale parameter of the Laplace distribution to be estimated from the data. For \ash{}, we set \texttt{mixcompdist="normal"} to use scale mixture of zero-mean Gaussians to approximate $g$.

\subsection{Pipeline for analyzing gene expression data}

Let $\theta_j$ denote the true $\log_2$-fold change in gene expression for each gene $j$. The analysis pipeline is used to provide, for each $\theta_j$, an estimate $X_j$ with a standard error $s_j$, such that $X_j$ can be assumed to be $N(\theta_j, s_j^2)$.

For RNA-seq data such as the mouse data, the analysis pipeline is described in Section \ref{sec:distortion}.

For microarray data such as the leukemia data, we use a widely-used analysis protocol implemented in the \limma{} software \citep{limma}. This yield an estimate $X_j$ for $\theta_j$, and a corresponding $p$-value $p_j$ from a moderated $t$-statistic \citep{smyth2004}. Then as in Section \ref{sec:distortion}, we convert the $p$-value to the corresponding $z$-score $z_j$ and use it to compute the effective standard deviation $s_j$.

\subsection{Reproducibility}

All the code generating the results and plots in this paper are available at \url{https://github.com/LSun/cashr_paper}.

The RNA-seq gene expression data from human liver tissues we use in this paper were generated by the GTEx Project, which was supported by the Common Fund of the Office of the Director of the National Institutes of Health, and by NCI, NHGRI, NHLBI, NIDA, NIMH, and NINDS. The data used for the analyses described in this paper were obtained from the GTEx Portal at \url{https://www.gtexportal.org}. In particular, the human liver RNA-seq data for realistic simulation are also available at \url{https://github.com/LSun/cashr_paper}.

The leukemia microarray data are available at \url{http://statweb.stanford.edu/~ckirby/brad/LSI/datasets-and-programs/datasets.html}.

The mouse heart RNA-seq data are available at \url{https://github.com/LSun/cashr_paper}.

\section{Representation of the correlated noise distribution} \label{sec:schwartzman2010}

If $Z$ are independent and $p$ is large then $F_p$ will be close to its mean, $\Phi$. This is guaranteed by well-established results like the Glivenko-Cantelli theorem and the Dvoretzky-Kiefer-Wolfowitz inequality \citep[e.g.][]{wasserman2006}. However, when $Z$ are correlated $F_p$ can be grossly different from $\Phi$, as we have seen in Section \ref{sec:distortion}. 
The covariance of $F_p$ indicates how far it tends to stray from its mean, $\Phi$, and therefore captures the extent of correlation-induced distortion.
\cite{schwartzman2010} provides the following elegant characterization of the covariance of $F_p$. For completeness we also put it here.

\begin{proposition} \citep[The mean, variance, and covariance functions of $F_p$;][]{schwartzman2010}

Assume $\forall i \neq j$, $\begin{bmatrix}Z_i \\ Z_j\end{bmatrix} \sim N\left(0, \begin{bmatrix}
1 & \rho_{ij} \\ \rho_{ij} & 1
\end{bmatrix}\right)$.
Let $\overline{\rho^l} \defeq \frac{1}{p\left(p - 1\right)}\sum\limits_{i, j : i \neq j}\rho_{ij}^l$. 
Then $\forall x, y \in \mathbb{R}$,
\begin{align}
E(F_p(x)) &= \Phi(x) \label{eq:mean_Fp}\\
var(F_p(x))
&=
\left(1 - \frac1p\right)
\sum\limits_{l = 1}^\infty \overline{\rho^l}\left[\frac{1}{\sqrt{l!}}\varphi^{(l - 1)}(x)\right]^2 + \frac1p\Phi(x)(1 - \Phi(x)) \label{eq:var_Fp} \\
cov(F_p(x), F_p(y)) &= \left(1 - \frac1p\right)
\sum\limits_{l = 1}^\infty \overline{\rho^l}
\left[\frac{1}{\sqrt{l!}}\varphi^{(l - 1)}(x)\right]
\left[\frac{1}{\sqrt{l!}}\varphi^{(l - 1)}(y)\right] + \frac1p[\Phi(\min(x, y)) - \Phi(x)\Phi(y)] 
\label{eq:cov_Fp}
\end{align}
\label{prop:cov_Fn}
\end{proposition}
\begin{proof}
The mean function is straightforward. The covariance function
\begin{align}
cov(F_p(x), F_p(y)) &= cov\left(\frac1p\sum\limits_{i = 1}^p \mathcal{I}(Z_i \leq x), \frac1p\sum\limits_{j = 1}^p \mathcal{I}(Z_j \leq y)\right) \nonumber
\\
&= E\left[\left(\frac1p\sum\limits_{i = 1}^p \mathcal{I}(Z_i \leq x)\right)\left(\frac1p\sum\limits_{j = 1}^p \mathcal{I}(Z_j \leq y)\right)\right] - E\left[\frac1p\sum\limits_{i = 1}^p \mathcal{I}(Z_i\leq x)\right]E\left[\frac1p\sum\limits_{j = 1}^p \mathcal{I}(Z_j \leq y)\right] \nonumber
\\
&=\frac1{p^2}\sum\limits_{i = 1}^p\sum\limits_{j = 1}^p E[\mathcal{I}(Z_i \leq x)\mathcal{I}(Z_j \leq y)] - \Phi(x)\Phi(y) \nonumber \\
&= \frac1{p^2}\sum\limits_{i = 1}^p\sum\limits_{j = 1}^p P(Z_i \leq x, Z_j \leq y) - \Phi(x)\Phi(y) \nonumber\\
&= \frac1{p^2}\sum\limits_{i\neq j}P(Z_i \leq x, Z_j \leq y)
+\frac1p \Phi(\min(x,y))
-\Phi(x)\Phi(y) \ .
\label{eq:cov_Fp_2}
\end{align}
According to Mehler's identity \citep{kibble1945}, under the assumption of $\{Z_i, Z_j\}$ being bivariate normal, the joint PDF can be written as
\begin{align}
p(x, y) = \varphi(x)\varphi(y)
+ \sum\limits_{l = 1}^\infty \rho_{ij}^l
\left[\frac{1}{\sqrt{l!}}\varphi^{(l)}(x)\right]
\left[\frac{1}{\sqrt{l!}}\varphi^{(l)}(y)\right] \ ,
\end{align}
so the joint CDF is
\begin{align}
P(Z_i \leq x, Z_j \leq y) = \Phi(x)\Phi(y)
+ \sum\limits_{l = 1}^\infty \rho_{ij}^l
\left[\frac{1}{\sqrt{l!}}\varphi^{(l - 1)}(x)\right]
\left[\frac{1}{\sqrt{l!}}\varphi^{(l - 1)}(y)\right] \ .
\label{eq:bvn_cdf}
\end{align}
\eqref{eq:cov_Fp_2} and \eqref{eq:bvn_cdf} lead to the covariance function \eqref{eq:cov_Fp}. Setting $x = y$ gives the variance function \eqref{eq:var_Fp}.
\end{proof}

Note that $var(F_p)$ has two parts. The second part $\frac1p\Phi(z)(1 - \Phi(z))$ is the familiar variance function when $Z$ are independent, and it quickly vanishes as $p$ increases. This is why $F_p$ of iid $N(0, 1)$ sample will not deviate much from $\Phi$ when $p$ is large. In contrast, the first part
\begin{align}
\left(1 - \frac1p\right)\sum\limits_{l = 1}^\infty \overline{\rho^l}\left(\frac{1}{\sqrt{l!}}\varphi^{(l - 1)}(x)\right)^2
\end{align}
demonstrates the effect of correlation. If $\overline{\rho^l}$ is non-negligible for large $p$, $var(F_p)$ will be non-vanishing, and so $F_p$ and the histogram of $Z$ are more likely to deviate substantially from $N(0,1)$.


When $p$ is large,
\begin{align}
cov(F_p(x), F_p(y)) \approx \sum\limits_{l = 1}^\infty \overline{\rho^l}
\left[\frac{1}{\sqrt{l!}}\varphi^{(l - 1)}(x)\right]
\left[\frac{1}{\sqrt{l!}}\varphi^{(l - 1)}(y)\right] \ .
\label{eq:cov_Fp_simple}
\end{align}
\eqref{eq:mean_Fp} and \eqref{eq:cov_Fp_simple} suggest we can characterize $F_p$ as \eqref{eq:Fbasis} \citep{schwartzman2010}, assuming $\overline{\rho^l} \ge 0$ for all $l \in \mathbb{N}$. This assumption should not be too demanding for large $p$ in practice. For example, when $l = 1$, \begin{align}
\overline{\rho} = \frac{1}{p(p - 1)}\sum\limits_{i \neq j}\rho_{ij} = \frac{1}{p(p - 1)}(\mathbf{1}^T\Sigma_Z\mathbf{1} - p) \geq \frac{1}{p(p - 1)}(-p) = -\frac{1}{p - 1} \ ,
\end{align}
following the fact that $\Sigma_Z$, the correlation matrix of $Z$, is positive semi-definite.

\bibliographystyle{cashr.bst}
\bibliography{cashr.bib}

\end{document}